%% file: main.tex
\documentclass[11pt,a4paper]{article}
\newlength{\frase}

\usepackage{a4}
\usepackage{latexsym}
\usepackage{amsmath,amssymb,amsthm}
\usepackage{xspace}
\usepackage{stmaryrd}

\usepackage{graphics}
\graphicspath{{./}{./FIGURE/}}

\usepackage{float}
\floatstyle{boxed} 
\restylefloat{figure}

\usepackage{a4}
\usepackage{latexsym}
\usepackage{amsmath,amssymb,amsthm}

\usepackage{pdfsync} 
\usepackage{color}
\usepackage[usenames,dvipsnames,svgnames,table]{xcolor}
\newcommand{\red}{\textcolor{red}}

\newcommand{\blue}{\textcolor{blue}}
 
\definecolor{light-gray}{gray}{0.95}

\theoremstyle{plain}
     \newtheorem{theorem}{Theorem}[section]
     \newtheorem{lemma}[theorem]{Lemma}
     \newtheorem{corollary}[theorem]{Corollary}
     \newtheorem{proposition}[theorem]{Proposition}
     
\theoremstyle{definition}
     \newtheorem{definition}[theorem]{Definition}

\theoremstyle{remark}
     \newtheorem{remark}{Remark}[section]
     \newtheorem{notation}{Notation}[section]

\DeclareSymbolFont{FormalScript}{U}{rsfs}{m}{n}
\DeclareSymbolFontAlphabet\mathscript{FormalScript}

\mathcode`\<="4268  
\mathcode`\>="5269  

\mathchardef\semicolon="603B 
\mathchardef\gt="313E
\mathchardef\lt="313C

\def\emptyset{\varnothing}
\def\epsilon{\varepsilon}
\def\theta{\vartheta}

\newcount\mscount
\input{prooftree}

\newcommand{\urule}[3]{%
   \prooftree #1 \justifies #2 \using #3 \endprooftree}
\newcommand{\brule}[4]{%
   \prooftree #1\ \ \ #2 \justifies #3 \using #4 \endprooftree}

\newcommand{\prova}[3]{%
   \prooftree #1 \justifies  #2  \using #3 \endprooftree}

\newcommand{\provasl}[3]{%
  \begin{array}[d]{c} \prooftree \mbox{$#1$} \justifies  #2  \thickness=0pt \using #3 \endprooftree\end{array}}

\newcommand{\provas}[3]{%
   \prooftree #1 \Justifies  #2  \thickness=0pt \using #3 \endprooftree}

\newcommand{\LT}[2]{%
   \prooftree #2 \leadsto #1 \endprooftree}

\def\conc{\conc}

\def\ltl{{\textsf{LTL}}}
\def\pltl{{\ltl}^{\textsf{P}}}

\def\KK{{\textsf{K}}}
\def\T{{\textsf{T}}}
\def\K4{{\textsf{K4}}}
\def\D{{\textsf{D}}}
\def\S4{{\textsf{S4}}}

\def\Q2{{\textsf{S4.2}}}

\def\red{\mathbf{\succ}}

\def\rred{\stackrel{*}{\red}}

\newcommand\pf[2]{{#1}^{#2}}

\newcommand\Ppf[2]{{#1}^{\langle #2\rangle}}

\def\forces{\models}

\newcommand\grado[1]{{\delta}[#1]}

\newcommand{\mix}[3]{\mathsf{Mix}(#1,#2)} 
 
\newcommand{\cana}[1]{#1\!-\!\pf{A}{\alpha}}

\newcommand{\rep}[2]{[#1 \Rsh  #2]}

\def\conc{\circ}

\def\toc{\mathcal{T}}
\def\pos{\toc^*}

\def\NN{\mathbb{N}}
\def\ZZ{\mathbb{Z}}

\newcommand{\less}[1]{\sqsubseteq_{#1}}
\newcommand{\lest}[1]{\sqsubset_{#1}}
\newcommand{\suc}[1]{\triangleleft_{#1}}

\newcommand{\iniz}[1]{\mathfrak{Init}[#1]}

\newcommand{\mdl}[1]{\mathcal{M}_{#1}}

\newcommand{\tmdl}[1]{\mathfrak{S}_{#1}}

\def\tree{\Theta}

\def\sys{\mathbb{M}}

\def\PBox{\blacksquare}
\def\PDiamond{\blacklozenge}
\def\Prev{\bullet}

\def\prop{\mathfrak{P}}

\def\mform{\mathfrak{mf}}
\def\pform{\mathfrak{pf}}

\def\R{\mathfrak{R}}

\newcommand{\logm}[1]{\mathfrak{M}[#1]}

\newcommand{\pfp}[2]{#1:#2}

\def\PFP{\mathcal{PFP}}

\newcommand{\commento}[1]{}

\title{A journey in modal proof theory:\\
From  minimal normal modal logic to discrete linear temporal logic
}

\author{S. Martini, A. Masini, M. Zorzi}
\date{}

\begin{document}
\maketitle
\pagestyle{myheadings}

\begin{quotation}
   \small \noindent {\sc Abstract:} %
Extending and generalizing the approach of 2-sequents (Ma\-si\-ni, 1992), we present sequent calculi for the classical modal logics in the \KK, \D, \T, \S4 spectrum. The systems are presented in a uniform way---different logics are obtained by tuning a single parameter, namely a constraint on the applicability of a rule. Cut-elimination is proved only once, since the proof  goes through independently from the constraints giving rise to the different systems. A sequent calculus for the discrete linear temporal logic \ltl\ is also given and proved complete. Leitmotiv of the paper is the formal analogy between modality and first-order quantification.
\end{quotation}
\thanks{\small{\it Mathematics Subject Classification (2000)}\/: 03B22, 03B45, 03F05.\\
\noindent{\it Keywords}:
  proof theory, sequent calculus, cut elimination, modal logic.}
%
%
\commento{
\fbox{ \blue{ per aiutare nella lettura e nella digitazione ecco una piccola tavola}}

\bigskip
{\scriptsize{
\blue{
\begin{tabular}{|c|c|l|}
\hline 
simbolo & significato &  latex \\ 
\hline 
$\conc$ & concatenazione stringhe & conc \\ 
\hline 
$\suc{-}$ & relazione  tra sequenze & suc[1] \\ 
\hline 
$\less{-}$ & relazione  tra sequenze • & less[1] \\ 
\hline 
$\lest{-}$ & relazione  tra sequenze  & lest[1] \\ 
\hline 
$\prop$ & simboli proposizionali & prop \\ 
\hline 
$\toc$ & token & toc \\ 
\hline 
$\pos$ & posizioni & pos \\ 
\hline 
$\mform$ & formule modali  & mform \\ 
\hline 
$\pform$ & formule posizionali  & pform \\ 
\hline 
$\rep{-}{-}$ & sostituzione & rep[2] \\ 
\hline 
$\logm{Z}$ & logica modale con assiomi Z & logm[1] \\ 
\hline 
$\sys$ & sistema modale & sys \\ 
\hline
$\alpha,\beta$ & posizioni logiche modali & • \\ 
\hline
$x,y,z$ & token & • \\ 
\hline
$s,t,u$ & posizioni logica LTL & • \\ 
\hline
$\iniz{\Gamma}$ & iniseme segmenti iniziali in $\Gamma$ & iniz[1] \\ 
\hline 
$\mdl{} $   &  modello di kripke  &  mdl[1] \\
\hline    
 $\tree$   &  albero  &  tree \\
\hline    
  $\R$  &  relazione di accessibilit\`a  & R  \\
\hline    
 $\tmdl{}$  &  struttura  & tmdl[1]  \\
\hline    
 $\rhd $  &   contrazione  & rhd   \\
\hline    
  $\red$  &  immediate reducibility  & red  \\
\hline    
  $\rred$  &   reducibility  & rred  \\
\hline    
 $\pfp{\alpha}{\beta}$  &  PFP   & pfp[2]  \\
\hline    
$\PFP$    &  set of Past-Present Positions  &  PFP  \\
\hline    
 $\PBox$   &    &  PBox \\
\hline    
  $\PDiamond$  &    &  PDiamond \\
\hline    
  $\Prev$  &    &  Prev \\
\hline    
\end{tabular} 
}
}
}
}
\tableofcontents
\input{introduction.tex}
\input{2sequents.tex}

\input{semantics.tex}

\input{discreteLTL.tex}
\input{pastLTL.tex}
\input{conclusions.tex}



\bigskip

\bibliographystyle{acm}
\bibliography{biblio} 
\end{document}

%% file: prooftree.tex
\newdimen\proofrulebreadth \proofrulebreadth=.05em
\newdimen\proofdotseparation \proofdotseparation=1.25ex
\newdimen\proofrulebaseline \proofrulebaseline=2ex
\newcount\proofdotnumber \proofdotnumber=3
\let\then\relax
\def\hfi{\hskip0pt plus.0001fil}
\mathchardef\squigto="3A3B
\newif\ifinsideprooftree\insideprooftreefalse
\newif\ifonleftofproofrule\onleftofproofrulefalse
\newif\ifproofdots\proofdotsfalse
\newif\ifdoubleproof\doubleprooffalse
\let\wereinproofbit\relax
\newdimen\shortenproofleft
\newdimen\shortenproofright
\newdimen\proofbelowshift
\newbox\proofabove
\newbox\proofbelow
\newbox\proofrulename
\def\shiftproofbelow{\let\next\relax\afterassignment\setshiftproofbelow\dimen0 }
\def\shiftproofbelowneg{\def\next{\multiply\dimen0 by-1 }%
\afterassignment\setshiftproofbelow\dimen0 }
\def\setshiftproofbelow{\next\proofbelowshift=\dimen0 }
\def\setproofrulebreadth{\proofrulebreadth}

\def\prooftree{
\ifnum	\lastpenalty=1
\then	\unpenalty
\else	\onleftofproofrulefalse
\fi
\ifonleftofproofrule
\else	\ifinsideprooftree
	\then	\hskip.5em plus1fil
	\fi
\fi
\bgroup
\setbox\proofbelow=\hbox{}\setbox\proofrulename=\hbox{}%
\let\justifies\proofover\let\leadsto\proofoverdots\let\Justifies\proofoverdbl
\let\using\proofusing\let\[\prooftree
\ifinsideprooftree\let\]\endprooftree\fi
\proofdotsfalse\doubleprooffalse
\let\thickness\setproofrulebreadth
\let\shiftright\shiftproofbelow \let\shift\shiftproofbelow
\let\shiftleft\shiftproofbelowneg
\let\ifwasinsideprooftree\ifinsideprooftree
\insideprooftreetrue
\setbox\proofabove=\hbox\bgroup$\displaystyle 
\let\wereinproofbit\prooftree
\shortenproofleft=0pt \shortenproofright=0pt \proofbelowshift=0pt
\onleftofproofruletrue\penalty1
}

\def\eproofbit{
\ifx	\wereinproofbit\prooftree
\then	\ifcase	\lastpenalty
	\then	\shortenproofright=0pt	
	\or	\unpenalty\hfil		
	\or	\unpenalty\unskip	
	\else	\shortenproofright=0pt	
	\fi
\fi
\global\dimen0=\shortenproofleft
\global\dimen1=\shortenproofright
\global\dimen2=\proofrulebreadth
\global\dimen3=\proofbelowshift
\global\dimen4=\proofdotseparation
\global\mscount=\proofdotnumber
$\egroup  
\shortenproofleft=\dimen0
\shortenproofright=\dimen1
\proofrulebreadth=\dimen2
\proofbelowshift=\dimen3
\proofdotseparation=\dimen4
\proofdotnumber=\mscount
}

\def\proofover{
\eproofbit 
\setbox\proofbelow=\hbox\bgroup 
\let\wereinproofbit\proofover
$\displaystyle
}%
\def\proofoverdbl{
\eproofbit 
\doubleprooftrue
\setbox\proofbelow=\hbox\bgroup 
\let\wereinproofbit\proofoverdbl
$\displaystyle
}%
\def\proofoverdots{
\eproofbit 
\proofdotstrue
\setbox\proofbelow=\hbox\bgroup 
\let\wereinproofbit\proofoverdots
$\displaystyle
}%
\def\proofusing{
\eproofbit 
\setbox\proofrulename=\hbox\bgroup 
\let\wereinproofbit\proofusing
\kern0.3em$
}

\def\endprooftree{
\eproofbit 
  \dimen5 =0pt
\dimen0=\wd\proofabove \advance\dimen0-\shortenproofleft
\advance\dimen0-\shortenproofright
\dimen1=.5\dimen0 \advance\dimen1-.5\wd\proofbelow
\dimen4=\dimen1
\advance\dimen1\proofbelowshift \advance\dimen4-\proofbelowshift
\ifdim	\dimen1<0pt
\then	\advance\shortenproofleft\dimen1
	\advance\dimen0-\dimen1
	\dimen1=0pt
	\ifdim  \shortenproofleft<0pt
        \then   \setbox\proofabove=\hbox{%
			\kern-\shortenproofleft\unhbox\proofabove}%
                \shortenproofleft=0pt
        \fi
\fi
\ifdim	\dimen4<0pt
\then	\advance\shortenproofright\dimen4
	\advance\dimen0-\dimen4
	\dimen4=0pt
\fi
\ifdim	\shortenproofright<\wd\proofrulename
\then	\shortenproofright=\wd\proofrulename
\fi
\dimen2=\shortenproofleft \advance\dimen2 by\dimen1
\dimen3=\shortenproofright\advance\dimen3 by\dimen4
\ifproofdots
\then
	\dimen6=\shortenproofleft \advance\dimen6 .5\dimen0
	\setbox1=\vbox to\proofdotseparation{\vss\hbox{$\cdot$}\vss}
	\setbox0=\hbox{%
		\kern\dimen6
		$\vcenter to\proofdotnumber\proofdotseparation
			{\leaders\box1\vfill}$%
		\unhbox\proofrulename}%
\else	\dimen6=\fontdimen22\the\textfont2 
	\dimen7=\dimen6
	\advance\dimen6by.5\proofrulebreadth
	\advance\dimen7by-.5\proofrulebreadth
	\setbox0=\hbox{%
		\kern\shortenproofleft
		\ifdoubleproof
		\then	\hbox to\dimen0{%
			$\mathsurround0pt\mathord=\mkern-6mu%
			\cleaders\hbox{$\mkern-2mu=\mkern-2mu$}\hfill
			\mkern-6mu\mathord=$}%
		\else	\vrule height\dimen6 depth-\dimen7 width\dimen0
		\fi
		\unhbox\proofrulename}%
	\ht0=\dimen6 \dp0=-\dimen7
\fi
\let\doll\relax
\ifwasinsideprooftree
\then	\let\VBOX\vbox
\else	\ifmmode\else$\let\doll=$\fi
	\let\VBOX\vcenter
\fi
\VBOX	{\baselineskip\proofrulebaseline \lineskip.2ex
	\expandafter\lineskiplimit\ifproofdots0ex\else-0.6ex\fi
	\hbox	spread\dimen5	{\hfi\unhbox\proofabove\hfi}%
	\hbox{\box0}%
	\hbox	{\kern\dimen2 \box\proofbelow}}\doll%
\global\dimen2=\dimen2
\global\dimen3=\dimen3
\egroup 
\ifonleftofproofrule
\then	\shortenproofleft=\dimen2
\fi
\shortenproofright=\dimen3
\onleftofproofrulefalse
\ifinsideprooftree
\then	\hskip.5em plus 1fil \penalty2
\fi
}


%% file: introduction.tex
\section{Introduction}\label{sec:introduction}
Proof theory of modal logic has always been a delicate subject---the ``intensionality'' of the modal connectives, even at the simple level of the normal logics based on the {\bf K} axiom, requires non-standard rules, both in  sequent calculi and in natural deduction systems.  In order to guarantee normalization, already in his seminal book~\cite{Prawitz:1965} Dag Prawitz is forced to formulate a natural deduction rule which has \emph{global} constraints: its applicability depends on the full structure of the proof tree rooted at the principal premise of the rule (and not only on the main connective of that premise and on the open assumptions of the tree(s), as it is the case for all the other rules, propositional or first-order). To treat modalities, several variants of the sequent format (or of natural deduction, or both) have been proposed: display calculi~\cite{Wans99}, hypersequents~\cite{Avron96,Negri:2011,CiRaWa14}, labelled systems~\cite{GabDeQ92,Simpson93, Vigano00a, Negri:2011} are just a few of them.  One of the authors of the present paper proposed in 1992 one of the earliest of these variants, called \emph{2-sequents}~\cite{Mas:TwoSeqProof:92,Mas:TwoSeqInt:93}, for the modal logic \D. The original 2-dimensional presentation (from which the system got the name) was later reformulated with a lighter syntax, using integer \emph{indexes} on formula occurrences, and extending it also to natural deduction. The simplicity of the approach made possible to tailor it  also to the intuitionistic case, and to apply it to the modalities (the ``exponentials'')  of linear logic (where indexes have a  natural interpretation in terms of ``box-nesting depth'')~\cite{MM:ComInt:95,MM:ONFineStr:95,GMM:an-exp-ext-seq}. The constraints on the applicability of modal rules are formulated by using only the indexes on the main premise and on the context (or on the open assumptions, in the case of natural deduction), thus having rules similar to the standard (propositional and first-order) ones. A distintive feature of the modal treatment in 2-sequents, is the formal analogy between necessitation and  universal quantification. Indeed, the introduction of necessity
\begin{quotation}
\noindent from $ \Gamma \vdash A$ infer $\Gamma \vdash \Box A$,
\end{quotation} 
which is sound only when all the formulas in $\Gamma$ are boxed, is the formal analog of the $\forall$-introduction rule
\begin{quotation}
\noindent from $ \Gamma \vdash A$ infer $\Gamma \vdash \forall x A$,
\end{quotation}
which is sound only when all the formulas in $\Gamma$ do not contain $x$ free. Indeed, this 
side condition may be read as: ``the formula $A$ must be independent, as far as $x$ is concerned, of the formulas in $\Gamma$''. The constraint on the $\Box$-introduction rule expresses a similar request of independence, which the 2-sequents allows to formulate also in analogous manner, as the absence of something from $\Gamma$ (see also~\cite{BarMas:apal} for a deeper discussion of the analogy.)

The present paper takes again this viewpoint and presents a general approach to modal proof-theory using 2-sequents, using the notion of \emph{position} of a formula occurrence (which generalises the concept of index that we used in our earlier work.)  While the previous papers treated only the cases of the classical \D, and the intuitionistic $\Box, \to, \wedge$-- fragments (no negation) of \D, \K4, \T, and {\S4}, we give here sequent calculi for all the normal, classical logics in the {\KK}, {\D}, {\T}, {\K4} and {\S4} spectrum.  The systems are presented in a uniform way---different logics are obtained by tuning a single parameter, namely the constraint on the applicability of the $\Box$-left rule (and $\Diamond$-right rule) in the various sequent calculi. Cut-elimination is proved only once, because the (standard!) proof techniques go through independently from the constraints of the different systems. 

Masini's 2-sequents are not the only variants on the sequent (or natural deduction) format which are based on annotations of formula occurrences. In most of them (e.g., notably Labelled Deductive Systems~\cite{GdQ92}, or Mints' Indexed Systems of Sequents---which mix sequents and tableaux,~\cite{Mints:1997}), however, annotations explicitly (and programmatically) reflect, in the formal proof calculus, the accessibility relation of the intended Kripke models. These approaches are successful in capturing a large array of different logics, and they allow, most of the time,  to prove general normalisation (or cut-elimination) results; see~\cite{Negri:2011} for a review of some of these approaches and their relations to the more standard, axiomatic presentations of modal theories. 

Our approach wants to stay at arm's length from these semantic considerations,  and  it builds instead, as we have already remarked, on the formal, inside-the-calculus notion of dependency of a formula from its premises. Of course, at the end some of the constraints of our systems will result similar to those of the other, more ``semantical'' approaches---this happens, however, as an a posteriori feature, which shows how the purely formal approach is able, in fact,  to reconstruct ``from below'' what other approaches assume in a top-down manner from semantical considerations. 

The paper develops in Section~\ref{sect:sequences} the proof theory for the classical logics in the {\KK}, {\D}, {\T}, {\K4} and {\S4} spectrum, proving cut-elimination (with a notion of subformula and, thus, consistency). It is also shown that the proposed systems prove all the theorems of the standard, axiomatic presentation of these logics. In order to prove the converse, Section~\ref{subs:seman} introduces a Kripke semantics for our 2-systems. The remaining sections of the paper are an exercise on the flexibility of our notion of positions. Section~\ref{sec:discreteLTL} gives a 2-sequent system for discrete linear temporal logic (\ltl), by generalising the notion of position and once again exploiting the analogy between quantifier rules and modal rules (where positions play the role of eigenvariables). The system is proved equivalent (by semantic means) to the usual axiomatic presentation of \ltl{}. Section~\ref{sect:past} shows a further generalisation, to deal also with  (unlimited) past.

%% file: 2sequents.tex
\section{Preliminary Notions}\label{sect:sequences}

As mentioned in the Introduction, formula occurrences will be labeled with \emph{positions}---sequences of uninterpreted \emph{tokens}. We introduce here the notation and operations that will be needed for such notions. 
Given a set $X$,  $X^*$ is the set of ordered finite sequences on $X$. With $<x_1,...,x_n>$ we denote the finite non empty sequence s.t. $x_1,\ldots,x_n \in X$; $<\ >$  is the empty sequence.

The  (associative) concatenation  of sequences $\conc:X^*\times X^*\to X^*$ is defined as 
\begin{itemize}
\item $<x_1,...,x_n>\conc <z_1,...,z_m> = <x_1,...,x_n,z_1,...,z_m>$,
\item $s\conc<\ > =<\ >\conc s= s$.
\end{itemize}

For $s\in X^*$ and $x\in X$, we sometimes write $s\conc x$ for $s\conc <x>$; and $x\in s$ as a shorthand for $\exists t,u\in X^*.\; s=t\conc<x>\conc u$.
The set $X^*$ is equipped with the following successor  relation

\[ s \suc{X} t \Leftrightarrow \exists x\in X.\; t=s\conc <x>\]

In the following
\begin{itemize}
\item $\suc{X}^0$  denotes the reflexive closure of $\suc{X}$;
\item $\lest{X}$  denotes the transitive closure of $\suc{X}$;
\item $\less{X}$  denotes the reflexive and transitive closure of $\suc{X}$;
\end{itemize}
\bigskip

%

Given three sequences $s,u,v\in X^*$ the \emph{prefix replacement} $s\rep{u}{v}$ is so defined 

$
s\rep{u}{v}=
\begin{cases}
v\conc t \quad \mbox{if\ } s=u\conc t
\\
s \quad\mbox{otherwise}
\end{cases}
$

When $u$ and $v$ have the same  length, the replacement is called \textit{renaming} of $u$ with $v$.

\section{2-sequent calculi}\label{sec:2seq}

The propositional modal language
${\cal L}$ contains the following symbols:
  \begin{enumerate}
  \item[--]  
    countably infinite \textit{proposition symbols}, $p_0,p_1,\ldots$;  
  \item[--] the \textit{propositional connectives} $\lor, \land, \to,
    \lnot;$
  \item[--]the \textit{modal operators} $\Box, \Diamond;$
  \item[--]the \textit{auxiliary symbols} $($\ and $).$
  \end{enumerate}

\begin{definition} The set $\mform$ of propositional \textit{modal formulas} 
  of ${\cal L}$ is the least set that contains the propositional
  symbols and is closed under application of the propositional
  connectives and the modal operators.
\end{definition}

In the following $\toc$ denotes a denumerable set of \emph{tokens}, ranged by meta-variables $x,y,z$, possibly indexed.
Let $\pos$ be the set the sequences on $\toc$ called \textit{positions}; meta-variables $\alpha,\beta,\gamma$  range on $\pos$, possibly indexed.

\begin{definition}\label{def:levfor} 
\begin{enumerate}
\item A  \textit{position-formula} (briefly \emph{p-formula}) is an expression of the form $\pf{A}{\alpha}$, where
  $A$ is a modal formula and $\alpha\in \pos$. We denote with $\pform$ the set of position formulas.
\item A \textit{2-sequent} is an
  expression of the form $\Gamma\vdash\Delta$, where $\Gamma$ and
  $\Delta$ are finite sequences of p--formulas.
\end{enumerate}
\end{definition}


Given a sequence $\Gamma$ of p-formulas, with $\iniz{\Gamma}$  we mean the set 
$\{ \beta : \exists  \pf{A}{\alpha} \in \Gamma.\; \beta \less{} \alpha \}$.

\medskip

\noindent \textit{Warning}: from now on we will  use the word ``sequent'' for  ``2-sequent'', when no ambiguity arises.

\bigskip

\subsection{A  class of normal modal systems}
We briefly recall  the axiomatic (``Hilbert-style'') presentation of normal modal systems.
Let $Z$ be a set of formulas. 
The normal modal logic $\logm{Z}$ is defined as smallest set $X$ of formulas verifying the following properties: 

\begin{description}
\item[(i)] $Z\subseteq X$
\item[(ii)] $X$ contains all instances of the following schemas:
\begin{description}
\item[1.]
$A\to(B \to A)$
\item[2.] $(A\to (B\to C))\to ((A\to B)\to (A\to C))$
\item[3.] $((\neg B\to \neg A)\to ((\neg  B\to A)\to B)) $
\item[K.] $\Box(A\to B)\to(\Box A \to \Box B)$
\end{description}
\item[MP]  if $A,A\to B \in X$ then $B\in X$;
\item[NEC] if $A\in X$ then $\Box A\in X$. 
\end{description}

We write $\vdash_{\logm{Z}}  A$  for  $A\in \logm{Z}$.
If $N_1,..,N_k$ are names of schemas, the sequence $N_1\ldots N_k$ denotes the set  
$[N_1]\cup...\cup [N_1]$, where  
$[N_i] =\{A: A \mbox{\ is an instance of the schema\ } N_i\}$.
Figure \ref{fig:modlogics} lists the standard axioms for the well-known modal systems ${\KK}$,  ${\D}$,  ${\T}$,  ${\K4}$,  ${\S4}$; we use $\sys$ as generic name for one of these systems.

\begin{figure}
\center
\begin{tabular}{c|c}

Axiom schema & Logic \\ 
\hline 
\begin{minipage}{30ex}
\begin{description}
\item[D] $\Box A \to \Diamond A$
\item[T] $\Box A \to A$
\item[4] $\Box A \to \Box\Box A$
\end{description}
\end{minipage}
 & 
\begin{minipage}{30ex}
	\begin{tabular}{l c l}
		\KK & $=$ & $\logm{\emptyset}$\\
		\D & $=$ & $\logm{\textbf{D}}$\\
		\T & $=$ & $\logm{\textbf{T}}$\\
		\K4 & $=$ & $\logm{\textbf{4}}$\\
		\S4 & $=$ & $\logm{\textbf{T, 4}}$
	\end{tabular}
\end{minipage} 
\end{tabular} 
\caption{Axioms for systems ${\KK}$,  ${\D}$,  ${\T}$,  ${\K4}$,  ${\S4}$}\label{fig:modlogics}
\end{figure}

\subsection{The sequent calculi $2_{{\KK}}, 2_{{\D}}, 2_{{\T}}, 2_{{\K4}}, 2_{{\S4}}$}

Figure \ref{fig:2s4rules} presents the 2-sequent calculus $2_{{\S4}}$, for the logic {\S4}. 
Observe that, as usual in sequent calculi presentations,  sequences of formulas ($\Gamma$, $\Delta$), or positions ($\alpha$, $\beta$) \emph{may be empty}, except when explicitly forbidden. 
The constraint on necessitation (rule $\vdash\Box$, and its dual $\Diamond\vdash$) is formulated as a constraint on position occurrences in the context, analogously to the usual constraint on variable occurrences for $\forall$-introduction.

Systems for other logics are obtained by restricting the application of some rules, using the positions present in the 2-sequents.   In particular, rules $\Box\vdash$ and $\vdash\Diamond$ are constrained for all the systems but $2_{{\S4}}$; moreover, for $2_{{\K4}}$ and $2_{{\KK}}$ also the \emph{cut}-rule is restricted. Figure~\ref{Fig-constraints} lists such constraints.

\begin{figure}
\center
\begin{tabular}{c|c}
\textbf{Calculus} & \textbf{Constraints on  the rules $\Box\vdash$ and $\vdash\Diamond$}\\ 
\hline 
$2_{{\S4}}$ & no constraints \\ 
\hline 
$2_{{\T}}$  & $\beta=<\ >$, or $\beta$ is a singleton sequence $<z>$ \\ 
\hline 
$2_{{\D}}$  & $\beta$ is a singleton sequence $<z>$  \\ 
 
\hline 
$2_{{\K4}}$  & $\beta$ is a non empty sequence;\\ &
there is at least a formula $\pf{B}{\alpha\conc\beta\conc\eta}$ in either $\Gamma$ or $\Delta$  \\ 

\hline
$2_{{\KK}}$  & $\beta$ is a singleton sequence $<z>$;\\ & 
there is at least a formula $\pf{B}{\alpha\conc\beta\conc\eta}$ in either $\Gamma$ or $\Delta$   \\ 
\multicolumn{2}{l}{\relax}\\
\multicolumn{2}{l}{\relax}\\

        & \textbf{Constraints on  the cut rule  }\\
\hline
$2_{{\D}}$, $2_{{\T}}$, $2_{{\S4}}$ & \mbox{\ no contraints\ }
\\
\hline
$2_{{\KK}}$, $2_{{\K4}}$ & $\alpha \in \iniz{\Gamma_1,\Delta_1}$ or  $\alpha \in \iniz{\Gamma_2,\Delta_2}$
	 \\ 	 
\end{tabular} 
\caption{Contraints}\label{Fig-constraints}
\end{figure}

Note that both $2_{{\K4}}$ and $2_{{\KK}}$, in addition to the constraint on the main position $\beta$, have also constraints on the context: in the modal rules $\Box\vdash$ and $\vdash\Diamond$ \textit{there must be another formula occurrence $\pf{B}{\alpha\conc\beta\conc\eta}$ in either $\Gamma$ or $\Delta$} (of course, $\alpha$ and/or $\eta$ may be empty). This prevents the derivation of $\pf{\Box A \to \Diamond A}{\gamma}$ (the p-formulas representing axiom {\textbf D}). 

\begin{remark}[On the cut rule for $2_{{\KK}}$, $2_{{\K4}}$]\mbox{}\\
The constraint is  necessary for $2_{{\K4}}$ and $2_{{\KK}}$, since it prevents the derivation of the unsound schema $\Ppf{\Diamond (A \to A)}{}$ (remember that {\KK} and {\K4} do not validate $\Diamond(\textbf{true})$). Indeed,  without the constraint we would have:

\def\aaa
{
\prova{\Ppf{A}{x}\vdash\Ppf{A}{x}}
{\vdash \Ppf{A\to A}{x}}
	{}
}
\def\aab
{
	\prova{
		\Ppf{A\to A}{x}\vdash\Ppf{A\to A}{x}
	}
	{
    \Ppf{A\to A}{x}\vdash\Ppf{\Diamond (A\to A)}{}
	}
	{}
}

	\[
	\prova{\aaa \aab}
	{
		\vdash\Ppf{\Diamond (A\to A)}{}
	}
	{}
	\]
It is easy to see that \emph{modus ponens} (from derivations of $\vdash \pf{A\to B}{\alpha}$ and $\vdash \pf{A}{\alpha}$, obtain a derivation of $\vdash \pf{B}{\alpha}$), which is necessary in order to prove the completeness of 2-systems, is derivable also in presence of this constraint.
\end{remark}

\begin{figure}[htb!]
\subsubsection*{Identity rules}
\begin{center}
  $\pf{A}{\alpha}\vdash\pf{A}{\alpha}\quad  Ax$ %
  \qquad\qquad %
  $\brule %
  {\Gamma_1\vdash  \pf{A}{\alpha},\Delta_1} %
  { \Gamma_2,\pf{A}{\alpha}, \vdash \Delta_2} %
  {\Gamma_1, \Gamma_2\vdash \Delta_1,\Delta_2 } %
  {\quad Cut } $ %
\end{center}

\subsubsection*{Structural rules}
\begin{center}
  $\urule %
  {\Gamma\vdash\Delta } %
  {\Gamma,\pf{A}{\alpha}\vdash\Delta } %
  {\quad W\vdash } $ %
\qquad\qquad
  $\urule %
  {\Gamma\vdash\Delta } %
  {\Gamma\vdash\pf{A}{\alpha},\Delta  } %
  {\quad\vdash W } $ %
\\ \bigskip
  $\urule %
  {\Gamma,\pf{A}{\alpha},\pf{A}{\alpha}\vdash\Delta  } %
  {\Gamma,\pf{A}{\alpha}\vdash\Delta  } %
  {\quad C\vdash } $ %
\qquad\qquad
  $\urule %
  {\Gamma\vdash\pf{A}{\alpha}, \pf{A}{\alpha},\Delta  } %
  {\Gamma\vdash\pf{A}{\alpha},\Delta  } %
  {\quad\vdash C } $ %
\\\bigskip
  $\urule %
  {\Gamma_1,\pf{A}{\alpha},\pf{B}{\beta},\Gamma_2\vdash \Delta } %
  {\Gamma_1,\pf{B}{\beta},\pf{A}{\alpha},\Gamma_2\vdash \Delta } %
  {\quad Exc\vdash } $ %
\qquad\qquad
  $\urule %
  {\Gamma \vdash \Delta_1,\pf{A}{\alpha},\pf{B}{\beta},\Delta_2 } %
  { \Gamma \vdash \Delta_1,\pf{B}{\beta},\pf{A}{\alpha},\Delta_2 } %
  {\quad\vdash Exc } $ %
\end{center}
\subsubsection*{Propositional rules}
\begin{center}
  $\urule %
  {\Gamma\vdash\pf{A}{\alpha},\Delta } %
  {\Gamma,\pf{\neg A}{\alpha}\vdash\Delta } %
  {\quad\neg\vdash} $ %
  \qquad\qquad
  $\urule %
  {\Gamma,\pf{A}{\alpha}\vdash\Delta } %
  {\Gamma\vdash\pf{\neg A}{\alpha},\Delta } %
  {\quad\vdash\neg } $ %

  \bigskip
  $\urule %
  {\Gamma,\pf{A}{\alpha}\vdash\Delta } %
  {\Gamma,\pf{A\land B}{\alpha}\vdash\Delta } %
  {\quad\land_1\vdash } $ %
\qquad\qquad
 $\urule %
  {\Gamma,\pf{B}{\alpha}\vdash\Delta } %
  {\Gamma,\pf{A\land B}{\alpha}\vdash\Delta } %
  {\quad\land_2\vdash } $ %

\bigskip

$\brule %
  {\Gamma_1\vdash\pf{A}{\alpha},\Delta_1 } %
  {\Gamma_2\vdash\pf{B}{\alpha},\Delta_2 } %
  {\Gamma_1,\Gamma_2\vdash\pf{A\land B}{\alpha},\Delta_1,\Delta_2 } %
  {\quad\vdash\land } $ %
\ 
$\brule %
  {\Gamma_1,\pf{A}{\alpha}\vdash\Delta_1 } %
  {\Gamma_2,\pf{B}{\alpha}\vdash\Delta_2 } %
  {\Gamma_1,\Gamma_2,\pf{A \lor B}{\alpha} \vdash\Delta_1,\Delta_2} %
  {\quad\lor \vdash } $ %

\bigskip

$\urule %
  {\Gamma\vdash\pf{A}{\alpha},\Delta  } %
  {\Gamma\vdash\pf{A\lor B}{\alpha},\Delta } %
  {\quad\vdash\lor_1 } $ %
\qquad\qquad
$\urule %
  {\Gamma\vdash\pf{B}{\alpha},\Delta  } %
  {\Gamma\vdash\pf{A\lor B}{\alpha},\Delta } %
  {\quad\vdash\lor_2 } $ %
\\ \bigskip
$\brule %
  {\Gamma_1,\pf{B}{\alpha}\vdash\Delta_1 } %
  {\Gamma_2\vdash\pf{A}{\alpha},\Delta_2 } %
  { \Gamma_1,\Gamma_2,\pf{A \to B}{\alpha}\vdash\Delta_1,\Delta_2} %
  {\quad\to \vdash } $ %
\qquad\qquad
$\urule %
  {\Gamma,\pf{A}{\alpha}\vdash\pf{B}{\alpha},\Delta  } %
  {\Gamma\vdash\pf{A\to B}{\alpha},\Delta } %
  {\quad\vdash\to } $ %
\end{center}
\subsubsection*{Modal rules}
\begin{center}
$\urule %
  {\Gamma,\pf{A}{\alpha\conc \beta}\vdash\Delta } %
  {\Gamma,\pf{\Box A}{\alpha}\vdash\Delta } %
  {\quad\Box\vdash } $ %
\qquad\qquad
$\urule %
  {\Gamma\vdash\pf{A}{\alpha\conc x},\Delta} %
  {\Gamma\vdash\pf{\Box A}{\alpha},\Delta } %
  {\quad\vdash\Box } $ %
\\ \bigskip
$\urule %
  {\Gamma,\pf{A}{\alpha\conc x}\vdash\Delta} %
  { \pf{\Gamma,\Diamond A}{\alpha}\vdash\Delta} %
  {\quad\Diamond\vdash } $ %
\qquad\qquad
$\urule %
  {\Gamma\vdash\pf{A}{\alpha\conc \beta},\Delta } %
  {\Gamma\vdash\pf{\Diamond A}{\alpha},\Delta} %
  {\quad\vdash \Diamond } $ %
\end{center}
\textbf{Constraints:} 

 In  rules $\vdash\Box$ and $\Diamond\vdash$, no position in $\Gamma,\Delta$ may start with $\alpha\conc x$; that is,  $\alpha\conc x\not\in \iniz{\Gamma,\Delta$}.

\caption{Rules for the System $2_{{\S4}}$}\label{fig:2s4rules}
\end{figure}

The position  $\alpha \conc x$ in the  rules $\vdash\Box$ and $\Diamond\vdash$ is the \emph{eigenposition} of that rule. 
It is well known that in standard first order sequent calculus  eigenvariables should be considered as bound variables. In particular, any eigenvariable in a derivation may always be substituted with a \emph{fresh} one (that is, a variable which does not occur in any other place in that derivation), without affecting the provable end sequent (up to renaming of its bound variables). Indeed, one may guarantee that each eigenvariable in a derivation is the eigenvariable of exactly one right $\forall$ or left $\exists$ rule (and, moreover, that variable occurs in the derivation only above the rule of which it is eigenvariable, and it never  occurs as a bound variable.) We will show analogous properties for the eigenpositions of 2-sequents, in order to define in a sound way a notion of \emph{prefix replacement for proofs} (that we defined at the end of Section~\ref{sect:sequences} for positions).  We denote with $\Gamma\rep{\alpha}{\beta}$ the obvious extension of prefix replacement to a sequence $\Gamma$ of p-formulas. The following lemmas allow the definition of a similar notion for proofs; the lemmas are valid for all the systems (that is, in presence of the constraints) of the table above.

\begin{lemma}
Let $\Pi$ be a 2-sequent proof with conclusion $\Gamma\vdash\Delta$, let $\delta\conc z$ be a position, and let $b$ be a fresh token (that is, not occurring in either $\Pi$ or $\delta\conc z$). Then we may define the \emph{prefix replacement} $\Pi\rep{\delta\conc z}{\delta\conc b}$,  a proof with conclusion 
$\Gamma\rep{\delta\conc z}{\delta\conc b}\vdash\Delta\rep{\delta\conc z}{\delta\conc b}$.
\end{lemma}
\begin{proof}
If $\Pi$ is an axiom $\pf{A}{\alpha}\vdash\pf{A}{\alpha}$, than $\Pi\rep{\delta\conc z}{\delta\conc b}$ is $\pf{A}{\alpha\rep{\delta\conc z}{\delta\conc b}}\vdash\pf{A}{\alpha\rep{\delta\conc z}{\delta\conc b}}$.\\
All inductive cases are trivial, except the modal rules.

If the last rule of $\Pi$ is 
$$
  \urule %
  {\Gamma\vdash\pf{A}{\alpha\conc x},\Delta} %
  {\Gamma\vdash\pf{\Box A}{\alpha},\Delta } %
  {\quad\vdash\Box }
$$
let $\Pi'$ be the subproof rooted at this rule. We have two cases, depending on whether the position $\delta\conc z$ is the eigenposition of the rule.
(i) If $\alpha\conc x = \delta\conc z$, obtain by induction the proof $\Pi'\rep{\alpha\conc x}{\alpha\conc b}$ with conclusion $\Gamma\vdash\pf{A}{\alpha\conc b},\Delta$ (remember that $\alpha\conc x\not\in\iniz{\Gamma,\Delta}$). Then $\Pi\rep{\delta\conc z}{\delta\conc b}$ is obtained from $\Pi'\rep{\alpha\conc x}{\alpha\conc b}$ by an application of $\vdash\Box$. 
(ii) If $\alpha\conc x \neq \delta\conc z$, obtain by induction the proof $\Pi'\rep{\delta\conc z}{\delta\conc b}$ with conclusion $\Gamma\rep{\delta\conc z}{\delta\conc b }\vdash\pf{A}{\alpha\rep{\delta\conc z}{\delta\conc b}\conc x},\Delta\rep{\delta\conc z}{\delta\conc b}$. Observe now that $\alpha\rep{\delta\conc z}{\delta\conc b}\conc x$ cannot be an initial segment of a formula in $\Gamma\rep{\delta\conc z}{\delta\conc b }, \Delta\rep{\delta\conc z}{\delta\conc b}$. 
Indeed, if for some $B^\gamma$ in $\Gamma,\Delta$ we had 
$\alpha\rep{\delta\conc z}{\delta\conc b}\conc x \less{}\gamma\rep{\delta\conc z}{\delta\conc b}$, since $b$ is fresh, this could only result from $\alpha\conc x$ being a prefix of $\gamma$, which is impossible.
Therefore, we may conclude with an application  of $\vdash\Box$, since its side-condition is satisfied.

If the last rule of $\Pi$ is 
$$
\urule %
  {\Gamma\vdash\pf{A}{\alpha\conc \beta},\Delta } %
  {\Gamma\vdash\pf{\Diamond A}{\alpha},\Delta} %
  {\quad\vdash \Diamond } 
$$ 
let, as before, $\Pi'$ be the subproof rooted at this rule and construct by induction the proof 
$\Pi'\rep{\delta\conc z}{\delta\conc b}$ with conclusion $\Gamma\rep{\delta\conc z}{\delta\conc b }\vdash\pf{A}{\alpha\conc \beta\rep{\delta\conc z}{\delta\conc b}},\Delta\rep{\delta\conc z}{\delta\conc b}$. It is easy to verify that any side condition of the $\vdash \Diamond$ rule (which depends on the specific system, according to the table above), is still verified after the prefix replacement. We may then conclude with a $\vdash \Diamond$ rule.

The left modal rules are analogous.
\end{proof}

By using the previous lemma, we obtain the following.
\begin{proposition}[eigenposition]\label{teo:eigenpos1}
Given a proof $\Pi$ of a sequent $\Gamma\vdash\Delta$, we may always find a proof $\Pi'$ ending with $\Gamma\vdash\Delta$ where all eigenpositions are distinct from one another.
\end{proposition}

Proof $\Pi'$ differs from $\Pi$ only for the names of positions.
In practice we will freely use  such a renaming all the times it is necessary (or, in other words, proofs are de facto equivalence classes modulo renaming of eigenpositions).
In a similar way to the previous lemmas we may obtain the following, which allows the prefix replacement of arbitrary positions (once eigenpositions are considered as bound variables, and renamed so that any confusion is avoided). When we use prefix replacement for proofs we will always assume that the premises of the following lemma are satisfied, implicitly calling for eigenposition renaming if this is not the case.

\begin{lemma}\label{lemma:prefix-replacement}
Let $\beta$ be an arbitrary position. Let $\Gamma\vdash\Delta$ be a provable sequent,  let $\delta\conc z$ be a position, and let $\Pi$ be a 2-sequent proof of $\Gamma\vdash\Delta$, where all eigenpositions are distinct from one another, and are different from $\delta\conc z$ and from $\beta$. Then we may define the \emph{prefix replacement} $\Pi\rep{\delta\conc z}{\delta\conc \beta}$,  a proof with conclusion 
$\Gamma\rep{\delta\conc z}{\delta\conc \beta}\vdash\Delta\rep{\delta\conc z}{\delta\conc \beta}$.
\end{lemma}



The notion of proof, provable sequent and   height $h(\Pi)$ of a proof $\Pi$ are  standard.

\begin{notation}
In order to simplify the graphical representation of $   $proofs, we will
use a double deduction line to indicate application of a rule preceded or followed
by a sequence of structural rules. So  we will write
$$\provas{\Gamma\vdash\Delta}{\Sigma\vdash \Theta}{r}$$ when  the sequent $\Sigma\vdash \Theta$ has
been obtained from $\Gamma\vdash\Delta$ by means of an  application of  rule $r$
and of a
finite number of structural rules.
\end{notation}

\subsection{2-sequents are complete}
We show in this section that the systems introduced in the previous section prove the same theorems of the Hilbert-style presentation of the corresponding logics:  if $\sys$ proves $A$, then $2_\sys$  proves $\vdash A^{<>}$. We start with the modal axioms; observe that the proof of each axiom satisfies the constraints on $\Box\vdash$ and $\vdash\Diamond$ of the corresponding 2-system. 

\subsubsection*{Axiom K}
\bigskip

\def\ka{
\prova{ 
  \pf{B}{<x>} \vdash \pf{B}{<x>}
  \quad
  \pf{A}{<x>} \vdash \pf{A}{<x>}
  }
  { 
  \pf{A}{<x>},  \pf{A\to B}{<x>}\vdash \pf{B}{<x>}
  }
  { \to\vdash } 
}

\def\kb{
\prova{ 
  \ka
  }
  { 
  \pf{A}{<x>},  \pf{\Box(A\to B)}{<\ >}\vdash \pf{B}{<x>}
  }
  { \Box\vdash } 
}

\def\kc{
\provas{ 
\kb	  }
  { 
    \pf{\Box A}{<\ >},  \pf{\Box(A\to B)}{<\ >}\vdash \pf{B}{<x>}
  }
  { \Box\vdash } 
}

\def\kd{
\prova{ 
  \kc			
  }
  { 
    \pf{\Box A}{<\ >},  \pf{\Box(A\to B)}{<\ >}\vdash \pf{\Box B}{<\ >}
  }
  {  \vdash \Box } 
}

\def\ke{
\provas{ 
  \kd
  }
  { 
  \pf{\Box(A\to B)}{<\ >} \vdash \pf{\Box A \to \Box B}{<\ >}
  }
  { \vdash\to } 
}

\def\kf{
\prova{ 
\ke 
  }
  { 
   \vdash \pf{\Box(A\to B)\to(\Box A \to \Box B)}{<\ >}
  }
  { \vdash\to } 
}

\[
\kf
\]

\subsubsection*{Axiom D}

\def\da{
\prova{ 
  \pf{A}{<x>}\vdash \pf{A}{<x>}
  }
  { 
    \pf{\Box A}{<\ >}\vdash \pf{A}{<x>}
  }
  { \Box\vdash } 
}

\def\db{
\prova{ 
  \da
  }
  { 
  \pf{\Box A}{<\ >}\vdash \pf{\Diamond A}{<\ >}
  }
  { \vdash\diamond } 
 } 

\def\dc{
\prova{ 
  \db			
  }
  { 
     \vdash \pf{\Box A \to \Diamond A}{<\ >}
  }
  { \vdash\to } 
}

$$
\dc
$$

\subsubsection*{Axiom T}
\def\ta{
\prova{ 
  \pf{A}{<\ >}\vdash \pf{A}{<\ >}
  }
  { 
    \pf{\Box A}{<\ >}\vdash \pf{A}{<\ >}
  }
  { \Box\vdash } 
}
\def\tb{
\prova{ 
  \ta
  }
  { 
  \pf{\Box A\to A}{<\ >}
  }
  { \vdash\to } 
}

$$
\tb
$$

\subsubsection*{Axiom 4}

\def\qa{
\prova{ 
  \pf{A}{<y,x>}\vdash \pf{A}{<y,x>}
  }
  { 
    \pf{\Box A}{<\ >}\vdash \pf{A}{<y,x>}
  }
  { \Box\vdash } 
}

\def\db{
\prova{ 
  \qa
  }
  { 
      \pf{\Box A}{<\ >}\vdash \pf{\Box A}{<y>}
  }
  { \vdash\Box } 
 } 

\def\dc{
\prova{ 
  \db			
  }
  { 
  \pf{\Box A}{<\ >}\vdash \pf{\Box \Box A}{<\ >}
  }
  { \vdash\Box } 
}

\def\de{
\prova{ 
  \dc		
  }
  { 
\vdash\pf{\Box A\to \Box \Box A}{<\ >}
  }
  { \vdash\to } 
}

$$
\de
$$

\bigskip
Closure under \textbf{GEN}  is obtained by showing that all positions in a provable sequent may be ``lifted'' by any prefix. Observe first that, for $\Gamma=A_1^{\gamma_1},\ldots, A_n^{\gamma_n}$,  we have $\Gamma\rep{<>}{\beta} = A_1^{\beta\conc\gamma_1},\ldots, A_n^{\beta\conc\gamma_n}$.
Finally, closure under \textbf{MP} is trivially obtained by means of the cut rule.
\begin{proposition}[lift]\label{lift:basics}
Let $\sys$ be one of the modal systems {\KK}, {\D}, {\T}, {\K4}, {\S4}, and let $\beta$ be a  position.
If 
$\Gamma\vdash \Delta$
is provable in $2_\sys$, so is the sequent 
$\Gamma\rep{<>}{\beta}\vdash \Delta\rep{<>}{\beta}$.
\end{proposition}
\begin{proof}
Like Lemma~\ref{lemma:prefix-replacement}: Standard induction on derivation (with suitable renaming of eigenpositions). It is easily verified that the constraints on the modal rules remain satisfied.
\end{proof}

\begin{corollary}
Let $\sys$ be one of the modal systems {\KK}, {\D}, {\T}, {\K4}, {\S4}.
\\
If $\vdash \pf{A}{<\ >}$ is provable in 
$2_\sys$ so is the sequent 
$\vdash \pf{\Box A}{<\ >}$.
\end{corollary}

\begin{theorem}[weak completeness]
Let $\sys$ be one of the modal systems {\KK}, {\D}, {\T}, {\K4}, {\S4}. If  $\vdash_\sys A$,  the sequent $\vdash \pf{A}{<\ >}$ is provable in $2_\sys$.
\end{theorem}

The converse of this theorem could be proved syntactically by a long and tedious work inside the axiomatic systems; instead, we will obtain it as Corollary~\ref{cor:interpr}, by a semantic argument.

\subsection{Cut elimination}\label{Sect:cut-elim}

We prove in this section the cut-elimination theorem for the 2-sequent systems we have introduced, adapting ideas
and techniques from~\cite{Girard:ptlc}.
We start with the standard notions of \emph{subformula} and \emph{degree}.

\begin{definition}[subformula]
  The set $Sub(\pf{A}{\alpha})$ of \textit{subformulas}  of a formula $\pf{A}{\alpha}$  is recursively  defined as follows:
  \begin{enumerate}
  \item[] $Sub(\pf{p}{\alpha})=\{\pf{p}{\alpha} \}$ if $p$ is a proposition symbol;
  \item[] $Sub(\pf{\lnot A}{\alpha}) = \{\pf{\lnot A}{\alpha} \} \cup Sub(\pf{A}{\alpha})$;
  \item[] $Sub(\pf{ A \# B}{\alpha}) = \{\pf{A \# B}{\alpha} \} \cup
    Sub(\pf{A}{\alpha})\cup Sub(\pf{B}{\alpha}), $ when $\#\in\{\to,\lor,\land\};$
  \item[] $Sub(\pf{\# A}{\alpha}) = \{\pf{\# A}{\alpha}  \} \cup
    \{Sub(\pf{A}{\alpha\conc \beta}): \beta\in P, \}$ 
    when $\#\in\{\Box,\Diamond\}.$
  \end{enumerate}
\end{definition}

\begin{definition}[degree] The \emph{degree} of modal formulas, p-formulas, and 2-sequent proofs are defined as follows.
\begin{enumerate}
\item
  The \emph{degree  of a modal formula} $A$,  $\deg(A)$,  is
  recursively  defined as:

  \begin{enumerate}
  \item $\deg(p) = 0$ if $p$ is a proposition symbol;
  \item $\deg(\lnot A) =  \deg(\Box A) = \deg(\Diamond A) = \deg(A)+1$;    
  \item $\deg(A \land B) = \deg(A\lor B) = \deg(A\to B) = \max\{\deg(A),\deg(B)\}+1$.
  \end{enumerate}
  
\item The \emph{degree of a p-formula} $\pf{A}{\alpha}$, $\deg(\pf{A}{\alpha})$, is just $\deg(A).$

\item The \emph{degree of a proof} $\Pi$, $\grado{\Pi}$, is the
  natural number   defined as follows: $$\grado{\Pi}=\left\{\begin{array}{cl} 0 & \mbox{if $\Pi$ is cut-free;}\\ \sup\{\deg(\pf{A}{\alpha})+1: \pf{A}{\alpha}\ \mbox{is a cut formula
      in}\ \Pi\}& \mbox{otherwise.}\end{array}\right.$$

\end{enumerate}
\end{definition}


Let $\Gamma$ be a sequence of formulas. We denote by $\cana{\Gamma}$ the
sequence obtained by removing all  occurrences  of $\pf{A}{\alpha}$
in $\Gamma.$ When writing  $\Gamma,\cana{\Gamma'}$ we actually mean
$\Gamma,(\cana{\Gamma'}).$
In the sequel, ordered pairs of natural numbers are intended to be lexicographically ordered. Hence one
can make proofs  by induction on pairs of numbers. The height $h(\Pi)$ of a proof $\Pi$ is defined in the usual way.

We will prove two different "mix lemmata", to take into account that the cut-rule for the systems $2_{{\KK}}$ and $2_{{\K4}}$ have special constraints, which are mirrored into the hypothesis of the lemma.

\begin{lemma}[Mix Lemma for $2_{{\D}}$, $2_{{\T}}$, $2_{{\S4}}$]\label{mainlemma} 
Let $\mathcal{S}$ be one of the systems $2_{{\D}}$, $2_{{\T}}$, $2_{{\S4}}$.
Let $n\in{\NN}$ and
  let $\pf{A}{\alpha}$ be a formula of degree $n$. Let now $\Pi$, $\Pi'$ be proofs of the sequents 
  $\Gamma\vdash \Delta$ and $\Gamma'\vdash \Delta'$, respectively,  satisfying the property $\grado{\Pi},\grado{\Pi'}\le n$. Then   one can obtain in an effective way from $\Pi$ and $\Pi'$ a  proof $\mix{\Pi}{\Pi'}{\pf{A}{\alpha}}$
  of the sequent $\Gamma,\cana{\Gamma'}\vdash\cana{\Delta},\Delta'$ satisfying the property $\grado{\mix{\Pi}{\Pi'}{\pf{A}{\alpha}}}\le n.$
\end{lemma} 
\begin{proof} The proof proceeds in a standard way, by induction on the pair $<h(\Pi),h(\Pi')>$. We highlight only the main points. 
 Let $\Pi$ and $\Pi'$ be
  $$\urule %
  {\left\{ %
      \provasl{ %
        \Pi_i } %
      { %
        \Gamma_i\vdash\Delta_i %
        } %
      {} %
    \right\}_{i\in I }} %
  {\Gamma\vdash\Delta } %
  {\  r }  %
  \qquad\mbox{and}\qquad
  \urule %
  {\left\{ %
      \provasl{ %
        \Pi'_j } %
      { %
        \Gamma'_j\vdash\Delta'_j %
        } %
      {} %
    \right\}_{j\in I' }} %
  {\Gamma'\vdash\Delta' } %
  {\ r' } $$ %
respectively, where $I$ and $I'$ are  $\emptyset$ (in case of an axiom), $
\{1\}$ or $\{1,2\}$.
We proceed by cases.
\begin{enumerate}
\item\label{ax} $r$ is $Ax.$  

If $\Gamma\vdash\Delta$ is $
  \pf{A}{\alpha}\vdash\pf{A}{\alpha},$ then one gets 
  $\mix{\Pi}{\Pi'}{\pf{A}{\alpha}}$ from $\Pi'$ by means of a suitable sequence of
  structural rules.

If $\Gamma\vdash\Delta$ is $
  \pf{B}{\beta}\vdash\pf{B}{\beta},$ for $B\ne A$ or  $\beta\ne \alpha,$ then one gets 
  $\mix{\Pi}{\Pi'}{\pf{A}{\alpha}}$  from $\Pi$ by a suitable sequence of structural rules. 

\item\label{axbis} $r'$ is $Ax.$\\ %
This case is symmetric to case~\ref{ax}.

\item\label{struct}$r$ is a structural rule.\\ %
Apply  induction hypothesis to the pair $<\Pi_1, \Pi'>$, then apply a suitable
  sequence of structural rules to get the conclusion.
  
\item\label{structbis} $r'$ is a structural rule\\ %
This case is symmetric to \ref{struct}.

\item\label{logrulenonA} $r$ is a cut or a logical rule not
  introducing   $\pf{A}{\alpha}$ to the right. \\ %
  Apply the induction hypothesis to each  pair $<\Pi_i,\Pi'>,$ 
   so obtaining the   proof
  $\mix{\Pi_i}{\Pi'}{\pf{A}{\alpha}},$ for $i\in I.$ The   proof
  $\mix{\Pi}{\Pi'}{\pf{A}{\alpha}}$ is then
$$\provas %
  {\left\{ %
      \provasl{ %
       \mix{\Pi_i}{\Pi'}{\pf{A}{\alpha}}
        } %
      { %
        \Gamma_i, \cana{\Gamma'} \vdash 
        \cana{\Delta_i} ,\Delta' %
        } %
      {} %
    \right\}_{i\in I }\vspace{0.1cm}} %
  {\Gamma, \cana{\Gamma'}\vdash\cana{\Delta},\Delta'  } %
  {r }  %
  $$

\item\label{logrulenonAbis} $r'$ is a cut or a logical rule
  not introducing
  $\pf{A}{\alpha}$ to  the left. \\ %
 This case is symmetric to~\ref{logrulenonA}.

\item\label{logruleA} $r$ is a logical rule introducing $\pf{A}{\alpha}$ to
  the right and $r'$ is a logical rule introducing $\pf{A}{\alpha}$ to the
  left.
  \begin{enumerate}
  \item $r$ is a propositional rule.\\ %
    This subcase is treated  as in the first order case  (see, for
    instance, \cite{Girard:ptlc}  or \cite{takeuti:pt}).

\item \label{box}$A$ is  $\Box B.$\\
  Let $\Pi$ and $\Pi'$ be
  $$
  \urule %
  {\provasl{\Pi_1}{\Gamma\vdash\pf{B}{\alpha\conc x},\Delta_1}{} } %
  {\Gamma\vdash \pf{A}{\alpha},\Delta_1} %
  { }
  \qquad\mbox{and}\qquad
  \urule %
  {\provasl{\Pi'_1}{\Gamma'_1,\pf{B}{\alpha\conc\beta}\vdash\Delta'}{} } %
  {\Gamma'_1,\pf{A}{\alpha}\vdash\Delta' } %
  { }
  $$respectively.
  Apply the induction hypothesis to the pairs of proofs 
  $<\Pi_1\rep{\alpha\conc x}{\alpha\conc\beta},\Pi'>$ and $<\Pi,\Pi'_1>,$ obtaining $\mix{\Pi_1\rep{\alpha\conc x}{\alpha\conc\beta}}{\Pi'}{\pf{A}{\alpha}}$ and
  $\mix{\Pi}{\Pi'_1}{\pf{A}{\alpha}}$, respectively. 
  The   proof $\mix{\Pi}{\Pi'}{\pf{A}{\alpha}}$ is then
  $${
    \provas %
    {  %
      \prova %
      {  %
        \provasl 
        {  %
          \mix{\Pi_1\rep{\alpha\conc x}{\alpha\conc\beta}}{\Pi'}{\pf{A}{\alpha}}
        } %
        {  %
          \Gamma,\cana{\Gamma'_1}\vdash\pf{B}{\alpha\conc\beta},\cana{\Delta_1},\Delta'
        } %
        {   } %
        \provasl %
        {  %
          \mix{\Pi}{\Pi'_1}{\pf{A}{k}}
        } %
        {  %
          \Gamma,\cana{\Gamma'_1},\pf{B}{\alpha\conc\beta}\vdash\cana{\Delta_1},\Delta'
        } %
        {   } %
      } %
      {  %
        \Gamma,\cana{\Gamma'_1},\Gamma,\cana{\Gamma'_1}
        \vdash
        \cana{\Delta_1},\Delta',\cana{\Delta_1},\Delta'
      } %
      { Cut  } %
    } %
    {  %
      \Gamma,\cana{\Gamma'_1} \vdash\cana{\Delta_1} ,\Delta'
    } %
    {   }
  }$$
  
\item  $A$ is $ \Diamond B.$
  This subcase  is symmetric to~\ref{box}.
\end{enumerate}
\end{enumerate}

In all cases, since the additional cuts are performed on
subformulas of $\pf{A}{\alpha},$ from the assumptions $\deg(\pf{A}{\alpha})=n$
and $\grado{\Pi},\grado{\Pi'}\le n$ we immediately get
$\grado{\mix{\Pi}{\Pi'}{\pf{A}{\alpha}}}\le n.$

\end{proof}

The above proof does not go through for the systems $ 2_{\KK}$ and $2_{\K4}$, because of the constraint on the context for the rules $\Box\vdash$ and $\vdash\Diamond$. Indeed, the case (5) of the proof would fail, as shown by the following two proof fragments.
Let $\alpha=\beta\conc x$ be the position of the statement of the lemma,
 $$
\urule %
{\provasl{\Pi_1}{
\vdash\pf{B}{\beta\conc x},    
\pf{A}{\beta\conc x}	
}{} } %
{
\vdash\pf{\Diamond B}{\beta},    
\pf{A}{\beta\conc x}
} %
{ }
\qquad\mbox{and}\qquad
\provasl{
\Pi'
}
{
\pf{A}{\beta\conc x}\vdash \pf{C}{\beta}
}{}
$$
If we apply the induction hypothesis to the pair $\langle \Pi_1, \Pi'\rangle$
we obtain 

$$
\provasl{\mix{\Pi_}{\Pi'}{}}
{ 
	\vdash \pf{B}{\beta\conc x}, \pf{C}{\beta}
}
{}	
$$
and at this time it is impossible to conclude with the $\vdash \Diamond$ rule, because via the induction hypothesis we deleted the only formula essential to validate the $\vdash \Diamond$ rule.

To fix the problem, we need a stronger statement of the lemma, which mirrors the constraint of the cut rule of $2_{{\KK}}$ and $2_{{\K4}}$.

\begin{lemma}[Mix Lemma for $2_{{\KK}}$, $2_{{\K4}}$]\label{mainlemma-bis} 
	Let $\mathcal{S}$ be one of the systems  $2_{{\KK}}$ or $2_{{\K4}}$.
	Let $n\in \NN$ and
	let $\pf{A}{\alpha}$ be a formula of degree $n$. Let now $\Pi,\Pi'$ be proofs of the sequents 
	$\Gamma\vdash \Delta$ and $\Gamma'\vdash \Delta'$, respectively,  satisfying the properties:
	\begin{itemize}
		\item $\grado{\Pi},\grado{\Pi'}\le n$;
		\item $\alpha \in \iniz{\Gamma, \cana{\Delta}}$, or $\alpha \in \iniz{\cana{\Gamma', \Delta'}}$
	\end{itemize}
	Then   one can obtain in an effective way from $\Pi$ and $\Pi'$ a  proof $\mix{\Pi}{\Pi'}{\pf{A}{\alpha}}$
	of the sequent $\Gamma,\cana{\Gamma'}\vdash\cana{\Delta},\Delta'$ satisfying the property $\grado{\mix{\Pi}{\Pi'}{\pf{A}{\alpha}}}\le n.$
\end{lemma} 

The proof is analogous to the proof of the previous lemma---it is readily seen that the hypotesis $\alpha \in \iniz{\Gamma, \cana{\Delta}}$, or $\alpha \in \iniz{\cana{\Gamma', \Delta'}}$ allows the conclusion also in  case (5). 

\def\xxa{\provasl{\Pi_1^*}{
		\Gamma_1 \vdash \pf{A}{\alpha}, \Delta'_1,  \pf{A}{\alpha}, \Delta''_1,
	}{}}

\def\xxb{\provas{\xxa}{\Gamma_1 \vdash \pf{A}{\alpha}, \Delta'_1,  \Delta''_1}{exch+contr}}

\def\xxc{\provas{\xxb}{\Gamma_1,\Gamma_2 \vdash \Delta'_1,  \pf{A}{\alpha}, \Delta''_1, \Delta_2}{exch+weak}}

\begin{theorem}[Cut elimination]\label{thm:cutels} 
Let $\sys$ be one of the modal systems $2_{{\KK}}$, $2_{{{\D}}}$, $2_{{\T}}$, $2_{{\K4}}$, and $ 2_{{\S4}}$. If $\Pi$ is a $2_\sys$--proof of $\Gamma\vdash\Delta$, then there
  exists  a   cut-free $2_\sys$--proof $\Pi^*$
  of $\Gamma\vdash\Delta.$ 
\end{theorem}
\begin{proof}
  By induction on the pair $<\grado{\Pi},h(\Pi)>.$   Suppose $\Pi$ is not cut-free and let $r$ be the last rule applied
  in $\Pi.$ We distinguish two cases:
 \begin{enumerate}
\item\label{rnotcut} $r$ is not a cut.\\
  Let $\Pi$ be
  $$\urule %
  {\left\{ %
      \provasl{ %
        \Pi_i } %
      { %
        \Gamma_i\vdash\Delta_i %
        } %
      {} %
    \right\}_{i\in I }} %
  {\Gamma\vdash\Delta } %
  {r, }
  $$
  where $I$ is one of $\{1\},$  $\{1,2\}$  Apply the 
  induction hypothesis to each $\Pi_i$, obtaining cut-free
proofs $\Pi^*_i,$ for $i\in I.$ A cut-free
 proof $\Pi^*$ of $\Gamma\vdash\Delta$ is then
  $$\urule %
  {\left\{ %
      \provasl{ %
        \Pi^*_i } %
      { %
        \Gamma_i\vdash\Delta_i %
        } %
      {} %
    \right\}_{i\in I }} %
  {\Gamma\vdash\Delta } %
  {r }
  $$
\item\label{riscut} $r$ is  a cut.\\
Let  $\Pi$ be
$$
\brule %
  {\provasl{\Pi_1}{\Gamma_1\vdash  \pf{A}{\alpha},\Delta_1}{}} %
  {\provasl{\Pi_2}{ \Gamma_2,\pf{A}{\alpha} \vdash \Delta_2}{}} %
  {\Gamma\vdash \Delta} %
  {\quad Cut }
$$  
  We have two subcases:
  \begin{enumerate}
  \item\label{case:cutS4}  $\sys$ is one of the  systems $ 2_{{{\D}}}, 2_{{\T}}, 2_{{\S4}}$:\\
   Apply the  induction hypothesis to $\Pi_1$ and $\Pi_2$  to obtain 
  cut-free proofs $\Pi^*_1$ and $\Pi^*_2$  of $\Gamma_1\vdash
  \pf{A}{\alpha},\Delta_1$ and $ \Gamma_2,\pf{A}{\alpha}\vdash \Delta_2$ respectively.

  Applying
  Lemma~\ref{mainlemma} to  
  the pair $<\Pi^*_1,\Pi^*_2>,$ one gets  a   proof
  $\Pi_0$
  of sequent 
  $\Gamma_1,\cana{\Gamma_2}\vdash\cana{\Delta_1},\Delta_2$ such that
  $\grado{\Pi_0}\le \deg(\pf{A}{\alpha})\lt\grado{\Pi}.$ 
  
  Finally  one  gets a  cut-free   proof of
  $\Gamma_1,\cana{\Gamma_2}\vdash\cana{\Delta_1},\Delta_2$  from $\Pi_0$ by
  induction hypothesis and, from it, a  cut-free   proof of
  $\Gamma\vdash\Delta$ by application  of a suitable sequence of structural rules.
  
  \item $\sys$ is one of the  systems $ 2_{{{\KK}}}, 2_{{\K4}}$:\\
  We have three subcases
  \begin{enumerate}
  	\item $\pf{A}{\alpha} \not\in \Delta_1$ and $\pf{A}{\alpha} \not\in \Delta_2$:
  	proceed as for case~\ref{case:cutS4}.
  	\item $\pf{A}{\alpha} \in \Delta_1$ :
  	Apply the  induction hypothesis to $\Pi_1$  to obtain 
  	cut-free proofs $\Pi^*_1$  of $\Gamma_1\vdash
  	\pf{A}{\alpha},\Delta_1$, 
  	then conclude in the following way:
  	\[
  	\xxc
  	\]
  	\item $\pf{A}{\alpha} \in \Delta_2$ :
  	simmetric to the previous one.
  \end{enumerate}

  \end{enumerate}

%
%
%
%
%
%
\end{enumerate}
\end{proof}


Let $\sys$ be one of the systems ${{\KK}}, {{\D}}, {{\T}}, {{\K4}}, {{\S4}}$. We have the following corollaries of cut-elimination.
\begin{corollary}[Subformula Property]\label{cor:subf}
  Each formula occurring in a cut-free   $2_\sys$--proof $\Pi$ is a
  subformula of some formula occurring in the conclusion of $\Pi.$
\end{corollary}

\begin{corollary}[Consistency]\label{cor:oscons}
  $2_\sys$ is consistent, namely there is  no   $2_\sys$--proof  of the empty sequent $\ \vdash\ $.
\end{corollary}

%% file: semantics.tex
\section{Semantics}\label{subs:seman}

We introduce in this section a tree-based  Kripke semantics for our 2-sequent modal systems, in order to prove their completeness with respect to the standard axiomatic presentations. 

\subsection{Trees}
Let $\NN^*$ be the set of finite sequences of natural numbers with the partial order $\less{\NN}$ as defined in Section~\ref{sect:sequences}.
%


\begin{definition}
A \textit{tree}  is  a subset  $\tree$ of $ \NN^*$  s.t. 
$<\ >\in \tree$; and
if $t\in \tree$ and $s \less{\tree} t$, then $s\in \tree, $ 
where $\less{\tree}$ is the restriction of $\less{\NN}$ to $\tree$.
\end{definition}
\noindent The elements of $\tree$ are called
\textit{nodes}; a \textit{leaf} is a node with no successors.
Given a tree $\tree$ and $s\in \tree$, we define $\tree_s$ (the \emph{subtree of $\tree$ rooted at $s$}) be the tree  defined
as:
$ s'\in \tree_s\ \Leftrightarrow\ s\conc s' \in \tree$.
Observe that 
$\tree_{<\ >} = \tree$.
In this section, $s$ and $t$ will range over the generic elements of $\Theta$.

\subsection{Tree-semantics}
If $At$ is the set of proposition symbols of our modal language, a \emph{Kripke model} is a triple $\mdl{}=<\tree,\nu, \R>$, where $\tree$ is a tree, $\nu:\tree\to 2^{At}$ is an assignment of proposition symbols to nodes, and  $\R\subseteq \tree\times\tree$.
Given a modal system $\sys\in{\{{\KK},{\D},{\T},{\K4},{\S4}\}}$, a \emph{$\sys$-model} is a Kripke model  $\mdl{_\sys}=<\tree,\nu, \R>$ s.t.
\bigskip

\begin{tabular}{|c|c|c|}
 \hline 
 \textbf{modal system} & \textbf{conditions on $\tree$} & \textbf{conditions on $\R$} \\ 
 \hline 
 {\KK}&  no condition  & $\R=\suc{\tree}$ \\ 
 \hline 
 {\D} & $\tree$ does not have leaves &  $\R=\suc{\tree}$  \\ 
 \hline 
 {\T} & no condition & $\R=\suc{\tree}^0$\\ 
 \hline 
 {\K4} & no condition  & $\R=\lest{\tree}$ \\ 
 \hline 
 {\S4} & no condition & $\R=\less{\tree}$  \\ 
 \hline 
 \end{tabular}  
 
\bigskip\noindent
The satisfiability (or forcing) relation of formulas on a Kripke model is standard; e.g., for a model $\mdl{}$ and node $s$,
$\mdl{},s \models \Box A \Leftrightarrow \forall t. s R t \Rightarrow \mdl{}, t\models A$.
As usual, we write $ \mdl{} \models A$, when $\mdl{},s \models  A$ for all nodes $s$ of $\mdl{}$.

\begin{theorem}[standard completeness]
For each modal system  $\sys$ in {\KK}, {\D}, {\T}, {\K4}, {\S4}, and for every formula $A$,
$\vdash_{\sys} A$ $\Leftrightarrow$  for all $\sys$--model $\mdl{}$, we have $\mdl{} \models A$.
\end{theorem}

\subsection{Semantics of 2--sequents}\label{semantics-sequents}
 \newcommand{\conv}[1]{{#1}\!\downarrow}
 
Let $\sys\in{\{{\KK}, {\D}, {\T}, {\K4}, {\S4}\}}$ be a modal system.
A $2_\sys$ \emph{structure} is a pair $\tmdl{\tree}=<\mdl{\tree},\rho>$ where:

\begin{itemize}
	\item $\mdl{\tree}$ is an $\sys$--model $<\tree,\nu, R>$
	\item $\rho:\pos \rightharpoonup \tree$ is a partial function from positions to nodes.
\end{itemize}
We write $\rho(x)\downarrow$ when the function $\rho$ is defined on input $x$. We require that $\rho(\alpha)\downarrow\; \Rightarrow \forall \beta \less{} \alpha.\rho(\beta)\downarrow$. 
Moreover, depending on the specific modal system, $\rho$ has to satisfy the following, additional constraints:
\bigskip

\begin{tabular}{|c|c|}
 \hline 
 \textbf{modal system} & \textbf{conditions on $\rho$} \\ 
 \hline 
 \KK&  $(\alpha\suc{\pos} \beta\ \&\  \conv{\rho(\alpha)}\ \&\  \conv{\rho(\beta)}) \Rightarrow  \rho(\alpha)\suc{\tree\ } \rho(\beta)$\\ 
 \hline 
 \K4 &  $(\alpha\suc{\pos} \beta\ \&\  \conv{\rho(\alpha)}\ \&\  \conv{\rho(\beta)}) \Rightarrow  \rho(\alpha)\lest{\tree } \rho(\beta)$ \\ 
 \hline 
 \D &  $ \rho \mbox{\ is total\ }\ \&\ (\alpha\suc{\pos} \beta \Rightarrow  \rho(\alpha)\suc{\tree\ } \rho(\beta))$  \\ 
 \hline 
 \T & $\rho \mbox{\ is total\ }\ \&\ (\alpha\suc{\pos} \beta \Rightarrow  \rho(\alpha)\suc{\tree\ }^0 \rho(\beta))$\\ 
 \hline 
 \S4 &  $ \rho \mbox{\ is total\ }\ \&\ ( \alpha\suc{\pos} \beta \Rightarrow  \rho(\alpha)\less{\tree}\rho(\beta))$  \\ 
 \hline 
 \end{tabular}  

\bigskip
\noindent
Since, in general, $\rho$ is partial (which is necessary for dealing with {\KK} and {\K4}), we need two different notions of satisfiability: $\models^\ell$ for the left hand side formulas in a sequent, and $\models^r$ for the right hand side formulas. Define then, for a  $2_\sys$ structure $<\mdl{\tree},\rho>$:
\begin{itemize}

\item
$\mdl{\tree},\rho\models^\ell \pf{A}{\alpha} \Leftrightarrow  (\conv{\rho(\alpha)} \ \&\ 	 \mdl{\tree},\rho(\alpha) \models A$);

\item  
$\mdl{\tree},\rho\models^r \pf{A}{\alpha} \Leftrightarrow  (\conv{\rho(\alpha)	}\ \Rightarrow \mdl{\tree},\rho(\alpha) \models A)$.
\end{itemize}

When \textit{$\rho$ is total} we observe that two notions of satisfiability collapse and we simply write: 
$$
\mdl{\tree},\rho\models \pf{A}{\alpha} \Leftrightarrow 
 \mdl{\tree},\rho(\alpha) \models A.$$
The definition is  extended to sequents:
$$
\mdl{\tree},\rho\models \Gamma\vdash\Delta \Leftrightarrow (
\forall \pf{A}{\alpha}\in\Gamma. \mdl{\tree},\rho\models^\ell \pf{A}{\alpha}
\Rightarrow \exists \pf{B}{\beta}\in\Delta. \mdl{\tree},\rho\models^r \pf{B}{\beta}
 ).
$$
Finally, given a modal system $\sys$,

$$
\Gamma\models_{\sys}\Delta \Leftrightarrow \forall\ 2_\sys\mbox{ structure\ } \tmdl{}, \tmdl{}\models \Gamma\vdash\Delta.
$$

We now introduce some notation for expressing substitution of values into the evaluation functions $\rho$, in correspondence of specific positions.
For $t\in\Theta$, define
\[
\rho\{\alpha\conc x/t\}(\beta)=
\begin{cases}
\rho(\beta) \mbox{\qquad \ if\ } \beta\neq \alpha\conc x \\
\rho(\alpha)\conc t \mbox{\quad otherwise\ }\\
\end{cases}
\]
As usual with expressions dealing with partial functions, any such substitution expression is undefined whenever it  formally contains an undefined 
subexpression; e.g., $\rho\{\alpha\conc x/\rho(\gamma)\}(\beta)$ is undefined when $\rho(\gamma)$ is undefined, or when $\rho(\alpha)\conc \rho(\gamma) \not\in \Theta$.

We define the following set of $\Theta$ elements:
\begin{itemize}
	\item $\Theta_{\KK}= \{t: |t|=1\}$;
	\item $ \Theta_{\K4}=\{t: |t|\gt0\}$;
	\item  $\Theta_{\D}=\{t: |t|= 1\} $;
	\item  $ \Theta_{\T}=\{t: |t|\leq 1\}$;
	\item  $ \Theta_{\S4}=\{t: |t|\geq 0\}$.
\end{itemize}
As for other notations, we will write $\Theta_\sys$ for any of these sets. We conclude with the crucial lemmas needed for the soundness of the modal rules. The first deals with the soundness of $\vdash\Box$ and $\Diamond\vdash$. 
\begin{lemma}\label{lemma:sub1} Let $\sys\in\{ {\KK}, {\D}, {\T}, {\K4}, {\S4}  \}$:
	\begin{enumerate}
		\item $\rho\models^r \Box \pf{A}{\alpha} \Leftrightarrow \forall t\in \Theta_\sys. \rho\{\alpha\conc x/t\}\models^r \pf{A}{\alpha\conc x};$
\item $\rho\models^r \Diamond \pf{A}{\alpha} \Leftrightarrow \exists t\in \Theta_\sys. \rho\{\alpha\conc x/t\} \models^r \pf{A}{\alpha\conc x}.$
	\end{enumerate}
\end{lemma}

The second lemma deals with the soundness of $\vdash\Diamond$ and $\Box\vdash$. 

\begin{lemma}\label{lemma:sub2}Let  $\sys\in\{ {\KK}, {\D}, {\T}, {\K4}, {\S4}  \}$:
		\item $\rho\models^r \pf{A}{\alpha\conc \beta}\Leftrightarrow \rho\{\alpha\conc x/\rho(\beta)\}\models^r\pf{A}{\alpha\conc x}$.
\end{lemma}

We are finally in the position to prove the \emph{soundness} theorem, by an easy induction on proofs which---we remark once again---strictly mimics the standard proof of soundness for the first order sequent calculus. 

\begin{theorem}[soundness]
	Let $\sys\in\{ {\KK}, {\D}, {\T}, {\K4}, {\S4}  \}$ be a modal system
If $\Gamma\vdash\Delta$ is derivable in $2_\sys$ then $\Gamma\models_{\sys}\Delta$.
\end{theorem}
\begin{proof}[Proof sketch.] By induction on the proof of $\Gamma\vdash\Delta$ in $2_\sys$. 
We  examine only the cases of $\vdash \Box$ and $\vdash\Diamond$.
\begin{description}
	\item[$\vdash\Box$] We observe first that the rule is the same for all the systems.

	$
	\forall \rho. \mdl{\tree},\rho\models \Gamma\vdash\pf{A}{\alpha\conc x},\Delta
	$
	\\
	$\Leftrightarrow$
	\\ 
	 $\forall \rho. \mdl{\tree},\rho\models^\ell \Gamma,\neg\Delta\Rightarrow 
	 \mdl{\tree},\rho\models^r \pf{A}{\alpha\conc x}$.
	 \\
	 $\Leftrightarrow$ (by the genericity of $\rho$)
	 \\ 
	 $\forall \rho \forall t\in\Theta_{\sys} \mdl{\tree},\rho\{\alpha\conc x/t\}\models^\ell \Gamma, \neg\Delta\Rightarrow 
	 \mdl{\tree},\rho\{\alpha\conc x/t\}\models^r \pf{A}{\alpha\conc x}$\\
	 $\Leftrightarrow$ (since $\alpha\conc x \not\in \iniz{\Gamma,\Delta}$)
	 \\  
	 $\forall \rho. \mdl{\tree},\rho\models^\ell \Gamma, \neg\Delta\Rightarrow  \forall t\in\Theta_\sys 
	 \mdl{\tree},\rho\{\alpha\conc x/t\}\models^r \pf{A}{\alpha\conc x}$.\\
	 Now Lemma~\ref{lemma:sub1} gives the conclusion.
	\item[$\vdash\Diamond$] The rule have different constraints in different systems; we deal with the $2_{\textbf{K4}}$ case, the others being similar or easier. 
\\
	$
 \mdl{\tree},\rho\models \Gamma\vdash\pf{A}{\alpha\conc \beta},\Delta
	$
	\\
	$\Leftrightarrow$
	\\  
	$\mdl{\tree},\rho\models^\ell \Gamma,\neg\Delta\Rightarrow 
	\mdl{\tree},\rho\models^r \pf{A}{\alpha\conc \beta}$.
	\\
	$\Leftrightarrow$ (by Lemma~\ref{lemma:sub2})
	\\
	$\mdl{\tree},\rho\models^\ell \Gamma,\neg\Delta\Rightarrow 
	\mdl{\tree},\rho\{\alpha\conc x/\rho(\beta)\}\models^r \pf{A}{\alpha\conc x}$
	\\
	Observe now that the side condition of $\vdash\Diamond$ for $2_{\K4}$ implies that  $\rho(\beta)\downarrow$. Lemma~\ref{lemma:sub1}(2) allows to conclude.
\end{description}
\end{proof}

\begin{corollary}\label{cor:interpr}
 If $\vdash\pf{A}{\alpha}$ is derivable in $2_\sys$, then in the Hilbert-style presentation of $\sys$ we have  $\vdash_{\sys} A$.
\end{corollary}

%% file: discreteLTL.tex
\section{Discrete Linear Temporal Logic}\label{sec:discreteLTL}
In the previous sections we have exploited the notions of position as sequence of tokens. The present section will explore what kind of modalities we may express when positions are treated as \emph{finite sets}.

\subsection{Relaxing positions: Towards linear time}
For the purpose of this section, positions are finite sets of tokens. Or, more precisely, we  quotient p-formulas with respect to the equivalence relation generated by the following schemas:
$\pf{A}{\alpha\conc x\conc y\conc  t  } \sim \pf{A}{\alpha\conc y\conc x\conc \beta}$, 
$\pf{A}{\alpha\conc x\conc x\conc \beta} \sim \pf{A}{\alpha\conc x\conc \beta} $.
Taking as base system the calculus $2_{\S4}$, the modal rules may be reformulated as:
\begin{center}
$\urule %
  {\Gamma,\pf{A}{\alpha\cup\beta}\vdash\Delta } %
  {\Gamma,\pf{\Box A}{\alpha}\vdash\Delta } %
  {\quad\Box\vdash } $ %
\qquad\qquad
$\urule %
  {\Gamma\vdash\pf{A}{\alpha\cup \{x\}},\Delta} %
  {\Gamma\vdash\pf{\Box A}{\alpha},\Delta } %
  {\quad\vdash\Box } $ %
\\ \bigskip
$
\urule %
  {\Gamma\pf{A}{\alpha\cup \{x\} } \vdash \Delta} %
  {\Gamma \Diamond \pf{A}{\alpha} \vdash\Delta } %
  {\quad\Diamond\vdash} 
$
\qquad\qquad
$\urule %
  {\Gamma\vdash\pf{A}{\alpha\cup \beta},\Delta } %
  {\Gamma\vdash\pf{\Diamond A}{\alpha},\Delta} %
  {\quad\vdash \Diamond } $ %
\end{center}

\textbf{Constraints:} In the  rules $\vdash\Box$ and $\Diamond\vdash$,  $x$ does not occur in any position in $\Gamma,\Delta$ (we write for this: $x\not\in\Gamma,\Delta$).

%
%
%

Let us call $2_{\Q2}$ the resulting sequent calculus.  It is easy to see that, indeed, the characteristic axiom of the modal system \Q2 is provable in $2_{\Q2}$:

\def\qda{
\prova{ 
  \pf{A}{\{ x,y\}} \vdash \pf{A}{\{ x,y \} }
  }
  { 
  \pf{\Box A}{\{ y \} }\vdash  \pf{A}{\{ x,y \} }
  }
  {  \Box\vdash } 
}

\def\qdb{
\prova{ 
  \qda
  }
  { 
  \pf{\Box A}{\{ y \} }\vdash  \pf{\Diamond A}{\{ x\} }
  }
  { \vdash\Diamond} 
}

\def\qdc{
\prova{ 
  \qdb
  }
  { 
  \pf{\Box A}{\{ y \} }\vdash  \pf{\Box\Diamond A}{\emptyset }
  }
  { \vdash\Box } 
}

\def\qdd{
\prova{ 
  \qdc
  }
  { 
  \pf{\Diamond\Box A}{\emptyset}\vdash  \pf{\Box\Diamond A}{\emptyset }
  }
  { \Diamond\vdash } 
}

\def\qde{
\prova{ 
  \qdd
  }
  { 
  \pf{\Diamond\Box A \to \Box\Diamond A}{\emptyset}
  }
  { \vdash\to } 
}
\[
\qde
\]
It is well known that this axiom is used to prove \Q2 complete for Kripke models (see Section~\ref{subs:seman}) whose accessibility relation is a directed partial order.
The following theorem follows by a tedious routine.\footnote{The theorem is not needed for the rest of the paper. The calculus  $2_{\Q2}$ is presented only as an intermediate step towards Linear Time Logic.}
\begin{theorem}
 $\vdash\pf{A}{\alpha}$ is derivable in $2_{\Q2}$  iff 
 $\vdash_{\Q2} A$.
\end{theorem}

%
%
%
%
%
%
%
%
%
%
%
%

\subsection{Axiomatic formulation of Linear Time Logic, \ltl}

The language of \ltl{}
  is a propositional  language with a
denumerable set $At$ of propositional letters, augmented with the
\textit{temporal}  operators
$\Box$ and $\circ.$ An axiomatization of \ltl{} (not a minimal one) is the following; we write $\vdash_{\ltl}$ for the provability relation in this system.

\bigskip
\noindent{\bf Axioms}
\begin{itemize}
\item[A0] All temporal instances of first order classical tautologies.
\item[A1]
  $\circ(A\to B)\to(\circ A\to\circ B)$
\item[A2] $\lnot\circ A\to\circ\lnot A$
\item[A3]
  $\Box(A\to B)\to(\Box A\to\Box B)$
\item[A4] $\Box A\to A$
\item[A5] $\Box A\to\Box\Box A$
\item[A6] $\Box A\to\circ A$
\item[A7] $\Box A\to\circ\Box A$
\item[A8] $ A\land\Box( A\to\circ A)\to\Box A$
\end{itemize}

\bigskip
\noindent{\bf Rules}

$$\prova{\vdash A\to B \quad \vdash A}{\vdash B}{\rm MP}\qquad\quad \prova{\vdash A}{\vdash\circ A}{\circ \rm G} \qquad\quad \prova{\vdash A}{\vdash\Box  A}{\Box \rm G}$$

\bigskip
From a semantical point of view, \ltl{} is complete with respect to Kripke models where the accessibility relations are discrete linear orders isomorphic to $\NN$.
Given the frame $\NN$ of natural numbers,  a
map $v:\NN\to 2^{At}$ and a natural number $m$, the relation of
satisfiability  by the 
model  $\NN_v=\langle{\NN }, v\rangle$ of a temporal formula $A$ at time $m$
(notation: $\NN_v\forces_m A$) is defined by induction on the complexity of
$A$ in the standard way. We recall here the definition for modalities.

\begin{enumerate}
\item[] $A$ is $\Box B$:\quad   ${\NN }_v\forces_m A\ \Leftrightarrow\ {\NN}_v\forces_n  B \mbox{\ for all\ } n\ge m;$ 
\item[] $A$ is $\circ  B$:\quad  ${\NN }_v\forces_m A\ \Leftrightarrow\ {\bf
  N}_v\forces_{m+1}  B.$ 
\end{enumerate}

\begin{remark}
The purpose of this paper is to provide a uniform proof theoretic treatment of the modal standard universal (namely  $\Box$) and existential (namely $\Diamond$) quantifiers, in various contexts, from the simplest modal logics (the minimal K system) to multimodal systems like \ltl.
For this reason, following~\cite{BarMas:apal, BaMaJANCL13, MVVJANCL2101, MVVJLC2011,BaMaAML2004}, we study only the \textsf{Until}-free fragment of \ltl. Indeed, \textsf{Until} is complex, as it is both existential and universal at the same time:  $A\ \textsf{Until}\ B$ holds at the current time instant $w$ iff either $B$ holds at $w$ or there exists an instant $w'$ in the future at which $B$ holds and such that $A$ holds at all instants between $w$ and $w'$. 
The treatment of \textsf{Until} would make us deviate significantly from the objectives of the paper and is left for further work.	
\end{remark}

The above axiom schemas  and rules are complete in a sense made precise by the
following (see \cite{emerson:handbook}):

\begin{theorem} For every temporal formula $A$
$$\vdash_{\ltl}A\ \Leftrightarrow\ 
{\NN }_v \forces_0 A \mbox{\ for all\ }v:{\NN }\to 2^{At}.$$
\end{theorem}

\subsection{Towards a sequent calculus}
An analysis of the axioms of the bimodal system \ltl{} makes clear that we must express both the behaviour of $\circ$ and $\Box$, \emph{per se}, and their mutual relations. 

\subsubsection{Next and Always}\label{Sec:next-always}

The treatment of the next operator $\circ$ is  simple---axioms $A1$, $A2$, and the inference rule $\circ G$ say that $\circ$ behaves as the \textit{necessity} operator of the modal system {\D}.
Moreover, axiom $A2$ says that $\circ$ is auto-dual and thus behaves both as the necessity and the possibly operator of  ${\D}$: 
no constraint is needed for the rule $\vdash\circ$. In the language of positions, this means that the positions for $\circ$ may be taken as the natural numbers (or, in other words, we can replace a position as a list of tokens with its length).
Consequently we have the rules:

\begin{center}
$\urule %
  {\Gamma,\pf{A}{n+1}\vdash\Delta } %
  {\Gamma,\pf{\circ A}{n}\vdash\Delta } %
  {\quad \circ \vdash } $ %
\qquad\qquad
$\urule %
  {\Gamma\vdash\pf{A}{n+1},\Delta} %
  {\Gamma\vdash\pf{\circ A}{n},\Delta } %
  {\quad\vdash\circ} $ %
\end{center}

As for the always operator $\Box$, we must at least  express ${\Q2}$, since \ltl{} is complete with respect to discrete total orders isomorphic to $\NN$, which trivially enjoy the property of directness. 
Consequently, positions and rules for $\Box$  must  inherit  those of $2_{\Q2}$.

\subsubsection{Interaction between $\circ$ and $\Box$}

Axioms $A6$, $A7$, and $A8$, tell us that ``semantically'' Always ($\Box$) behaves as the reflexive and transitive closure of Next ($\circ$).
To formalise this interaction between $\circ$ and $\Box$,  we extend the notion of position. Roughly speaking a position becomes a pair---one component is a set of tokens, needed to handle $\Box$ and its interaction to $\Diamond$ and $\circ$; the other component is a natural number, for handling $\circ$.

\subsection{The calculus $2_{\ltl}$}

\begin{definition}The set of \textit{positions for \ltl{}} is the set of 
   pairs $<n,S>$ where $n$ is a natural number and $S$ is a finite set of tokens
   from a
   denumerable set\linebreak $T=\{x_0,x_1,\ldots\}.$ 
\end{definition}

Let $s=<n,S>$ and $t=<m,T>$ be positions. For the sake of simplicity we introduce the following notation:

\begin{itemize}
\item $s\oplus t$ for $ <n+m,S\cup T>;$
\item if $T=\emptyset$, we write $s\oplus m$ for $s\oplus t$;
\item if $t=<0,\{x\}>$, we write $s\oplus x$ for $s\oplus t$;
\item if $t=<n,\{ \}>$, we abbreviate $t$ with $n$.
\item we let $s[t/x] =
\left\{
\begin{array}{cl}
<n+m,(S\setminus\{x\})\cup T> &\ \mbox{if}\ x\in S;\\
s & \ \mbox{otherwise.}
\end{array}
\right.$
\end{itemize}

The rules for Next ($\circ$) are the ones already discussed in Section~\ref{Sec:next-always}, acting only on the second component of positions:

\begin{center}
$\urule %
  {\Gamma,\pf{A}{s\oplus 1}\vdash\Delta } %
  {\Gamma,\pf{\circ A}{s}\vdash\Delta } %
  {\quad \circ \vdash } $ %
\qquad\qquad
$\urule %
  {\Gamma\vdash\pf{A}{s\oplus 1},\Delta} %
  {\Gamma\vdash\pf{\circ A}{s},\Delta } %
  {\quad\vdash\circ} $ %
\end{center}

Regarding $\Box$ and its dual $\Diamond$ (which must be introduced anyway, since we are in a classical setting), the question is more delicate. We already discussed why $\Diamond\vdash$ and $\vdash\Box$  are, \textit{de facto}, those for \Q2: we formulate them as operating on the  first component of  positions:

\begin{center}
$\urule %
  {\Gamma\vdash\pf{A}{s \oplus  x},\Delta} %
  {\Gamma\vdash\pf{\Box A}{s},\Delta } %
  {\quad\vdash\Box } $ %
\qquad \qquad
$\urule %
  {\Gamma,\pf{A}{s \oplus  x}\vdash\Delta} %
  { \pf{\Gamma,\Diamond A}{s}\vdash\Delta} %
  {\quad\Diamond\vdash} $ %
\end{center}
with the proviso that $x\not\in\Gamma,\Delta$.

The more delicate axioms $A6$ and $A7$ force to operate on both components of positions:

\begin{center}
$\urule %
  {\Gamma,\pf{A}{s \oplus t}\vdash\Delta } %
  {\Gamma,\pf{\Box A}{s}\vdash\Delta } %
  {\quad\Box\vdash } $ %
\qquad\qquad
$\urule %
  {\Gamma\vdash\pf{A}{s \oplus t},\Delta } %
  {\Gamma\vdash\pf{\Diamond A}{s},\Delta} %
  {\quad\vdash \Diamond } $ %
\end{center}

\bigskip

\noindent These rules, however, do not validate axiom $A8$: $ A\land\Box( A\to\circ A)\to\Box A$. It is easy to see that $A8$ is the formal analogue (in the temporal setting) of an induction axiom for natural numbers (see the soundness theorem stated in Section~\ref{Sect:SemInduction} for the details).
A possible formulation of induction in a sequent calculus for \textbf{PA} (see e.g.~\cite{takeuti:pt}) is the following:
\[
\urule{
\Gamma,A(x) \vdash A(x+1), \Delta
}
{
\Gamma,A(0) \vdash A(t), \Delta
}
{\rm ind}
\]
where $x\not\in \Gamma,\Delta$. Following, once again, our formal analogy between first-order variables and positions, we may express axiom $A8$ by means of the following rule:

\[
\urule{
\Gamma, \pf{A}{s \oplus x} \vdash  \pf{A}{s \oplus x\oplus 1} , \Delta
}
{
\Gamma,\pf{A}{s} \vdash \pf{A}{s\oplus t}, \Delta
}
{\rm IND}
\]
where $x\not\in s, \Gamma,\Delta$.

Figure~\ref{fig:temprule} summarises the full set of rules of System $2_{\ltl}$.
 
	\begin{figure}[htb!]
\subsubsection*{Identity rules, Structural rules, Propositional rules }
		Those of the systems of Section~\ref{fig:2s4rules}, formulated with the new notion of position.
		
\subsubsection*{Temporal rules}
		\begin{center}
			$\urule %
			{\Gamma,\pf{A}{s\oplus t}\vdash\Delta } %
			{\Gamma,\pf{\Box A}{ s }\vdash\Delta } %
			{\quad\Box\vdash } $ %
			\qquad\qquad
			$\urule %
			{\Gamma\vdash\pf{A}{ s \oplus x},\Delta} %
			{\Gamma\vdash\pf{\Box A}{ s },\Delta } %
			{\quad\vdash\Box } $ %
			\\ \bigskip
			$\urule %
			{\Gamma,\pf{A}{ s \oplus x}\vdash\Delta} %
			{ \pf{\Gamma,\Diamond A}{ s }\vdash\Delta} %
			{\quad\Diamond\vdash } $ %
			\qquad\qquad
			$\urule %
			{\Gamma\vdash\pf{A}{ s \oplus t},\Delta } %
			{\Gamma\vdash\pf{\Diamond A}{ s },\Delta} %
			{\quad\vdash \Diamond } $ \\
		\bigskip
		$\urule %
		{\Gamma,\pf{A}{ s \oplus 1}\vdash\Delta} %
		{ \pf{\Gamma,\circ A}{ s }\vdash\Delta} %
		{\quad\circ \vdash } $ %
		\qquad\qquad
		$\urule %
		{\Gamma\vdash\pf{A}{ s \oplus 1},\Delta } %
		{\Gamma\vdash\pf{\circ A}{ s },\Delta} %
		{\quad\vdash \circ } $ \\
			\bigskip
			$\urule{
				\Gamma, \pf{A}{s \oplus x} \vdash  \pf{A}{s \oplus x\oplus 1} , \Delta
			}
			{
				\Gamma,\pf{A}{s} \vdash \pf{A}{s\oplus t}, \Delta
			}
			{\rm IND}
			$
		\end{center}
		\textbf{Constraints:} 
		
		In  rules $\vdash\Box$, $\Diamond\vdash$ and $\rm IND$,  $x \not\in s,\Gamma,\Delta $.

		\caption{Rules for the System $2_{\ltl}$}\label{fig:temprule}
	\end{figure}
	
	
 \subsection{Weak completeness}
 We show here 
  that the system $2_{\ltl}$ proves the same theorems of \ltl, i.e.  if $\ltl$ proves $A$, then $2_{\ltl}$  proves $\vdash A^{s}$ for a generic position $s$ (in particular $s=0$).

  The proof of axioms A1, A3 A4 and A5, is identical (up to the use of the new notion of positions) to the ones given for axioms \textbf{K}, \textbf{T} and \textbf{4}.

  
  \subsubsection*{Axioms A1, A3}
  \bigskip

   \def\kaa{
   	\prova{ 
   		\pf{B}{s\oplus 1} \vdash \pf{B}{s\oplus 1}
   		\quad
   		\pf{A}{s\oplus 1} \vdash \pf{A}{s\oplus 1}
   	}
   	{ 
   		\pf{A}{s\oplus 1},  \pf{A\to B}{s\oplus 1}\vdash \pf{B}{s\oplus 1}
   	}
   	{ \to\vdash } 
   }
   
   \def\kbb{
   	\prova{ 
   		\kaa
   	}
   	{ 
   		\pf{A}{s\oplus 1},  \pf{\circ(A\to B)}{s}\vdash \pf{B}{s\oplus 1}
   	}
   	{ \circ\vdash } 
   }
   
   \def\kcc{
   	\provas{ 
   		\kbb	  }
   	{ 
   		\pf{\circ A}{s},  \pf{\circ(A\to B)}{s}\vdash \pf{B}{s\oplus 1}
   	}
   	{ \circ\vdash } 
   }
   
   \def\kdd{
   	\prova{ 
   		\kcc			
   	}
   	{ 
   		\pf{\circ A}{s},  \pf{\circ(A\to B)}{s}\vdash \pf{\circ B}{s}
   	}
   	{  \vdash \circ } 
   }
   
   \def\kee{
   	\provas{ 
   		\kdd
   	}
   	{ 
   		\pf{\circ(A\to B)}{s} \vdash \pf{\circ A \to \circ B}{s}
   	}
   	{ \vdash\to } 
   }
   
   \def\kff{
   	\prova{ 
   		\kee 
   	}
   	{ 
   		\vdash \pf{\circ(A\to B)\to(\circ A \to \circ B)}{s}
   	}
   	{ \vdash\to } 
   }
   
   \def\ka{
   	\prova{ 
   		\pf{B}{s\oplus x} \vdash \pf{B}{s\oplus x}
   		\quad
   		\pf{A}{s\oplus x} \vdash \pf{A}{s\oplus x}
   	}
   	{ 
   		\pf{A}{s\oplus x},  \pf{A\to B}{s\oplus x}\vdash \pf{B}{s\oplus x}
   	}
   	{ \to\vdash } 
   }
   
   \def\kb{
   	\prova{ 
   		\ka
   	}
   	{ 
   		\pf{A}{s\oplus x},  \pf{\Box(A\to B)}{s}\vdash \pf{B}{s\oplus x}
   	}
   	{ \Box\vdash } 
   }
   
   \def\kc{
   	\provas{ 
   		\kb	  }
   	{ 
   		\pf{\Box A}{s},  \pf{\Box(A\to B)}{s}\vdash \pf{B}{s\oplus x}
   	}
   	{ \Box\vdash } 
   }
   
   \def\kd{
   	\prova{ 
   		\kc			
   	}
   	{ 
   		\pf{\Box A}{s},  \pf{\Box(A\to B)}{s}\vdash \pf{\Box B}{s}
   	}
   	{  \vdash \Box } 
   }
   
   \def\ke{
   	\provas{ 
   		\kd
   	}
   	{ 
   		\pf{\Box(A\to B)}{s} \vdash \pf{\Box A \to \Box B}{s}
   	}
   	{ \vdash\to } 
   }
   
   \def\kf{
   	\prova{ 
   		\ke 
   	}
   	{ 
   		\vdash \pf{\Box(A\to B)\to(\Box A \to \Box B)}{s}
   	}
   	{ \vdash\to } 
   }
   
  \[
  \kff \qquad \kf
  \]
  
%
%
%
%
%
  
\subsubsection*{Axioms A4 and A5}

  \def\ta{
  	\prova{ 
  		\pf{A}{s}\vdash \pf{A}{s}
  	}
  	{ 
  		\pf{\Box A}{s}\vdash \pf{A}{s}
  	}
  	{ \Box\vdash } 
  }
  \def\tb{
  	\prova{ 
  		\ta
  	}
  	{ 
  		\pf{\Box A\to A}{s}
  	}
  	{ \vdash\to } 
  }

  \def\qa{
  	\prova{ 
  		\pf{A}{s\oplus y\oplus x}\vdash \pf{A}{s\oplus y\oplus x}
  	}
  	{ 
  		\pf{\Box A}{s}\vdash \pf{A}{s\oplus y\oplus x}
  	}
  	{ \Box\vdash } 
  }
  
  \def\db{
  	\prova{ 
  		\qa
  	}
  	{ 
  		\pf{\Box A}{s}\vdash \pf{\Box A}{s\oplus y}
  	}
  	{ \vdash\Box } 
  } 
  
  \def\dc{
  	\prova{ 
  		\db			
  	}
  	{ 
  		\pf{\Box A}{s}\vdash \pf{\Box \Box A}{s}
  	}
  	{ \vdash\Box } 
  }
  
  \def\de{
  	\prova{ 
  		\dc		
  	}
  	{ 
  		\vdash\pf{\Box A\to \Box \Box A}{s}
  	}
  	{ \vdash\to } 
  }
  
  $$
   \tb \qquad \de
  $$

\subsubsection*{Axioms A2, A6, A7}

\def\ze{
	\provas{\pf{A}{s\oplus 1}\vdash \pf{A}{s\oplus 1}}{  \vdash \pf{A}{s\oplus 1}, \pf{\neg A}{s\oplus 1}}{ \vdash\neg }	
}
\def\zd{
	\provas{\ze}{  \vdash \pf{\circ A}{s}, \pf{\neg A}{s\oplus 1}}{ \vdash\circ }	
}

\def\zc{
	\provas{\zd}{  \vdash \pf{\circ A}{s}, \pf{ \circ\neg A}{s}}{ \vdash\circ }	
}

\def\zb{
	\prova{\zc}{ \pf{\neg\circ A}{s} \vdash \pf{ \circ\neg A}{s}}{\neg \vdash }	
}

\def\za{
\prova{\zb}{\vdash \pf{\neg\circ A\to \circ\neg A}{s}}{\vdash \to}	
}


\def\uc{
	\prova{ 
		\pf{A}{s\oplus 1}\vdash \pf{A}{s\oplus 1}
	}
	{ 
		\pf{A}{s\oplus 1}\vdash \pf{\circ A}{s}
	}
	{ \vdash \circ} 
}

\def\ub{
	\prova{ 
		\uc
	}
	{ 
		\pf{\Box A}{s}\vdash \pf{\circ A}{s}
	}
	{ \Box\vdash } 
}
\def\ua{
	\prova{ 
		\ub
	}
	{ 
		\pf{\Box A\to \circ A}{s}
	}
	{ \vdash\to } 
}

\def\wdd{
\prova{
	\pf{A}{s\oplus x \oplus 1} \vdash \pf{A}{s\oplus x \oplus 1}
}{\pf{\Box A}{s} \vdash \pf{A}{s\oplus x \oplus 1}}{\Box\vdash }	
}

\def\wc{
	\prova{ 
	\wdd
	}
	{ 
		\pf{\Box A}{s}\vdash \pf{\Box A}{s\oplus 1}
	}
	{ \vdash \Box} 
}

 \def\wb{
	\prova{ 
		\wc
	}
	{ 
		\pf{\Box A}{s}\vdash \pf{\circ \Box A}{s}
	}
	{ \vdash \circ } 
}
\def\wa{
	\prova{ 
		\wb
	}
	{ 
		\pf{\Box A\to \circ \Box A}{s}
	}
	{ \vdash\to } 
}

$$
\zb \quad \wa \quad \ua
$$

 
\subsubsection*{Axiom IND}
 
\def\aaa{
\prova{
 \pf{A}{s\oplus x \oplus 1} \vdash \pf{A}{s\oplus x \oplus 1} 
 }
 {
 \pf{\circ A}{s\oplus x} \vdash \pf{A}{s\oplus x \oplus 1} 
 }
 {\circ\vdash}
} 
 
\def\bbb{
\prova{
\aaa\quad \pf{A}{s\oplus x} \vdash \pf{A}{s\oplus x} 
}
{
\pf{A\to \circ A}{s\oplus x }, \pf{A}{s\oplus x} \vdash \pf{A}{s\oplus x \oplus 1} 
}
{\to\vdash}
 }
 
 \def\ccc{
 \prova{
 \bbb
 }
 {
 \pf{\Box(A\to \circ A)}{s} , \pf{A}{s\oplus x} \vdash \pf{A}{s\oplus x \oplus 1} 
 }
 { \Box\vdash}
 }
 
\def\ddd{
\prova{
  \ccc
  }
  {
  \pf{\Box(A\to \circ A)}{s} , \pf{A}{s} \vdash \pf{A}{s\oplus z} 
  }
  { IND}
}
 
\def\eee{
\prova{
  \ddd
  }
  {
  \pf{A\wedge \Box(A\to \circ A)}{s} \vdash \pf{A}{s\oplus z} 
  }
  { \wedge\vdash }
}

\def\fff{
\prova{
  \eee
  }
  {
  \pf{A\wedge \Box(A\to \circ A)}{s} \vdash \pf{\Box A}{s}
  }
  { \vdash\Box}
 }
 
\[
\prova{
  \fff
  }
  {
  \vdash \pf{A\land\Box( A\to\circ A)\to\Box A}{s}
  }
  { \vdash\to }
\]

\begin{proposition}[lift]
If the sequent $\vdash \pf{A}{s}$ is provable in $2_{\ltl}$, so is the sequent $\vdash \pf{A}{s\oplus t}$,  for each position $t$.
\end{proposition}
\begin{proof}
As Proposition~\ref{lift:basics}, proving a general lifting property for sequents $\Gamma\vdash\Delta$ provable in $2_{\ltl}$, and exploiting renaming of eigenpositions. 
\end{proof}

As an immediate  consequence we have that

\begin{corollary}

If  $\vdash \pf{A}{<0,\emptyset >}$ is provable  
 so is the sequent 
$\vdash \pf{\Box A}{<0,\emptyset >}, \vdash \pf{\circ A}{<0,\emptyset >}$.
\end{corollary}

Finally we have:
%

\begin{theorem}[weak completeness]
If $\vdash_{\ltl} A$, then the sequent $\vdash \pf{A}{<0,\emptyset >}$ is provable in $2_{\ltl}$.
\end{theorem}
\bigskip

\subsection{Semantics and soundness}\label{Sect:SemInduction}
For the sake of simplicity, we take a version of $2_{\ltl}$ with an explicit version of  axiom $A8$, instead of rule $\rm IND$.
Let $2^i_{\ltl}$ be the system where  rule $\rm IND$ is replaced with the following schema, at any position $s$:
\[
{\pf{A\land\Box( A\to\circ A)\to\Box A}{s}}\qquad {IndAx}
\]
It is  simple to see that $2{\ltl}$ and $2^i_{\ltl}$ derives the same sequents.

\def\ia{
\prova{\iba \quad \ic}{
\Gamma,\pf{A}{s} \vdash \pf{A}{s\oplus t}, \Delta	
}
{\mbox{\tiny\it cut}}
}

\def\iba{
	\prova{\ibbs \quad \ibc}{
		\Gamma,\pf{A}{s} \vdash \pf{\Box A}{s}, \Delta	
	}
	{\mbox{\tiny\it cut}}	
}

\def\ibc{
\LT{
	\pf{A}{s}, \pf{\Box(A \to \circ A)}{s} \vdash \pf{\Box A}{s}
}{
\vdash \pf{A\land\Box( A\to\circ A)\to\Box A}{s}
}
{
\mbox{}
}
}

\def\ibbs{   
   \prova
   {
   \ibb
   }
   {\Gamma\vdash\pf{\Box(A \to \circ A)}{s},\Delta
   }
   {\vdash\Box}
}

\def\ibb{
	\prova
	{
	\ibba	
	}
{
\Gamma\vdash\pf{A\to \circ A}{s\oplus x}, \Delta	
}
{\vdash \to}	
}

\def\ibba{
	\prova
	{
\Gamma, \pf{A}{s\oplus x}\vdash \pf{A}{s\oplus x\oplus 1}, \Delta
	}
	{
		\Gamma, \pf{A}{s\oplus x} \vdash\pf{\circ A}{s\oplus x}, \Delta	
	}
	{\vdash \circ}	
}

\def\ic{
	\prova{\pf{A}{s\oplus t} \vdash \pf{A}{s\oplus t}}{
		\pf{\Box A}{s} \vdash \pf{A}{s\oplus t}
	}
	{\Box \vdash}	
}

\begin{proposition}	
$2_{\ltl}$ derives $\Gamma\vdash \Delta$ iff $2^i_{\ltl}$ derives $\Gamma\vdash \Delta$.
\end{proposition}
\begin{proof}
	The ``only if part'' has already been  proved above.
	As for the ``if part'', the following derivation shows how to derive the conclusion of the $\rm IND$ rule from its premise and an instance of $IndAx$.
	\[
\ia
	\]
	\end{proof}

We now define the semantics of system $2_{\ltl}$ relative to the frame
$\langle {\NN },+,0,1\rangle.$ 
Models based on  $\langle {\NN },+,0,1\rangle$  are of the form $${\NN }_{a,v}=\langle {\NN},+,0,1,a:T\to{\NN },v:{\NN }\to 2^{At}\rangle.$$
In ${\NN }_{a,v}$ one can assign a value $a(s)\in {\NN }$ to each
position $s=<n,S>$: $$a(s)=n+\sum_{x\in S}a(x).$$

\begin{definition} The model ${\NN }_{a,v}$ {\em satisfies} the position formula $\pf{A}{s}$ (notation: ${\NN }_{a,v}\models{\pf{A}{s}}$) iff  ${\NN }_v\forces_{a(s)}A.$
\end{definition}

The semantics of position formulas is thus reduced to the standard semantics of \ltl, for example:
\begin{itemize}
\item ${\NN }_{a,v}\models \pf{\circ A}{s}$ iff 
	${\NN }_v\forces_{a(s)}\circ A$ iff
	${\NN }_v\forces_{a(s)+1} A$;
\item ${\NN }_{a,v}\models \pf{\Box A}{s}$ iff 
	${\NN }_v\forces_{a(s)}\Box A$ iff 
	$\forall n\geq 0. {\NN }_v\forces_{a(s)+n} A$;
\item ${\NN }_{a,v}\models \pf{\Diamond A}{s}$ iff 
	${\NN }_v\forces_{a(s)}\Diamond A$ iff 
	$\exists n\geq 0. {\NN }_v\forces_{a(s)+n} A$.
\end{itemize}
We extend the definition to sequents: 
$$
\NN_{a,v} \models \Gamma\vdash\Delta \Leftrightarrow (
\forall \pf{A}{s}\in\Gamma.\NN_{a,v}\models \pf{A}{s}
\Rightarrow \exists \pf{B}{t}\in\Delta.\NN_{a,v}\models \pf{B}{t}
 )
$$
and finally
$$
\Gamma\models_{2_{\ltl}}\Delta \Leftrightarrow \forall a,v.  \NN_{a,v}, \models \Gamma\vdash\Delta.
$$

As usual, towards soundness we need a substitution lemma. 

\begin{lemma}\label{lemma:subLTL} 
	\begin{enumerate}
	\item Let $x\not\in s$, then 
		${\NN}_{a,v} \models \Box \pf{A}{s} \Leftrightarrow 
		\forall n\in \NN. {\NN}_{a[x/n],v} \models   \pf{A}{s\oplus x}.$
	\item Let $x\not\in s$, then ${\NN}_{a,v} \models  \pf{\Diamond A}{s} \Leftrightarrow 
		\exists n\in \NN. {\NN}_{a[x/n],v} \models   \pf{A}{s\oplus x}.$
	\item ${\NN}_{a,v} \models   \pf{A}{s\oplus t}
		\Leftrightarrow
		{\NN}_{a[x/a(t)],v} \models   \pf{A}{s\oplus x}$
	\end{enumerate}
\end{lemma}

The proof of the soundness theorem proceeds as in the standard case, but for temporal induction.
\begin{theorem}[soundness]
If $\Gamma\vdash\Delta$ is derivable in $2^i_{\ltl}$, then $\Gamma\models_{2_{\ltl}}\Delta$.
\end{theorem}

\begin{proof}[Proof sketch.]
	We examine here only the cases of $\vdash \Box$,  $\vdash\Diamond$, and $IndAx$.
	\begin{description}
		\item[$\vdash\Box$] \mbox{}\\
		$\forall a,v. \NN_{a,v}\models \Gamma\vdash\pf{A}{s\oplus x},\Delta$
		\\
		$\Leftrightarrow$\\
	$\forall a,v. \NN_{a,v}\models \Gamma,\neg\Delta \Rightarrow  \NN_{a,v}\models \pf{A}{s\oplus x}$
	\\
	$\Leftrightarrow$\\
	$\forall n, a, v. \NN_{a[x/n],v}\models \Gamma,\neg\Delta \Rightarrow  \NN_{a[x/n],v}\models \pf{A}{s\oplus x}$
	\\
	$\Leftrightarrow$ (since $x\in \Gamma,\Delta$)\\
	$\forall n, a, v. \NN_{a,v}\models \Gamma,\neg\Delta \Rightarrow  \NN_{a[x/n],v}\models \pf{A}{s\oplus x}$
	\\
	$\Leftrightarrow$\\
	$\forall n, a, v. \NN_{a,v}\models \Gamma,\neg\Delta \Rightarrow  \NN_{a,v}\models \pf{\Box A}{s}$
	\\
	$\Leftrightarrow$\\
	$\forall a,v. \NN_{a,v}\models \Gamma\vdash\pf{\Box A}{s},\Delta$

		\item[$\vdash\Diamond$] \mbox{}\\
		$\forall a,v. \NN_{a,v}\models \Gamma\vdash\pf{A}{s\oplus t},\Delta$
		\\
			$\Leftrightarrow$\\
			$\forall a,v. \NN_{a,v}\models \Gamma,\neg\Delta \Rightarrow  \NN_{a,v}\models \pf{A}{s\oplus t}$
			\\
			$\Leftrightarrow$\\
			$\forall a,v. \NN_{a,v}\models \Gamma,\neg\Delta \Rightarrow  \NN_{a[x/a(t)],v}\models \pf{A}{s\oplus x}$
			\\
			$\Rightarrow$\\
			$\forall a,v. \NN_{a,v}\models \Gamma,\neg\Delta \Rightarrow \exists m \NN_{a[x/m],v}\models \pf{A}{s\oplus x}$
			\\
			$\Leftrightarrow$\\
			$\forall a,v. \NN_{a,v}\models \Gamma,\neg\Delta \Rightarrow \NN_{a,v}\models \pf{\Diamond A}{s}$
			\\
			$\Leftrightarrow$\\
				$\forall a,v. \NN_{a,v}\models \Gamma\vdash\pf{\Diamond A}{s},\Delta$
		\item[indAx]		We need to prove that
		\[
		\forall v, a. \NN_{a,v}\models \pf{A\land\Box( A\to\circ A)\to\Box A}{s}.
		\]
Let us consider the set $[A]^k_v=\{n: \NN_v\models_{k+n} A  \}$.
With simple calculations we have that
	\[
	\forall v, a . \NN_{a,v}\models \pf{A\land\Box( A\to\circ A)\to\Box A}{s}
	\]
iff
	\[
	\forall v, a(0\in [A]^{a(s)}_v \& (x\in [A]^{a(s)}_v\Rightarrow (x+1)\in [A]^{a(s)}_v)
	\Rightarrow \forall y(y\in [A]^{a(s)}_v ).
	\]
\end{description}
\end{proof}

\begin{remark}
When $IndAx$ is instantiated to a propositional symbol $p$ for the generic formula $A$, it is immediate to see that
\[	
\forall v, a. \NN_{a,v}\models \pf{p\land\Box( p\to\circ p)\to\Box p}{0}
\]
iff
\[
\forall x (0\in [p]^{a(s)}_v \& (x\in [p]^{a(s)}_v\Rightarrow (x+1)\in [p]^{a(s)}_v)
\Rightarrow \forall y(y\in [p]^{a(s)}_v )
\]
iff
\[
\forall v.
0\in v^{-1}\{p\} \& \forall x (x\in v^{-1}\{p\} \Rightarrow x+1\in v^{-1}\{p\})
\Rightarrow \forall x. x\in v^{-1}\{p\}
\]
iff, for the genericity of $v$
\[
\forall S \subseteq \NN.
0\in S \& \forall x (x\in S \Rightarrow x+1\in S)
\Rightarrow \forall x. x\in S.
\]
In other words, in order to establish the soundness for \ltl{}, we have to assume full (second order) induction on the natural numbers.
\end{remark}	

\subsection{On cut and induction}
It is well known that an induction rule is a big obstacle for a full cut elimination (namely, a cut elimination with a subformula principle), or at least for a cut elimination with cut-rank bounded by a fixed integer number, usually called \emph{partial} cut elimination (see, for instance, the discussion in~\cite{Girard:ptlc}, pp. 123--125.)
This phenomenon is well known for \textbf{PA}, where a partial cut elimination would lead to a consistency proof of \textbf{PA} inside \textbf{PA} itself, thus contradicting the second incompleteness theorem.

From a  combinatorial point of view, the problem in proving cut-elimination is that permutative cuts are blocked by the induction rule.
 Even if $2_{\ltl}$ is (apparently) a weak logical system, it exhibits the same phenomenon. Indeed, let us consider the following proof:
 
 \def\x{
 \prova{
 	\pf{\circ p}{x\oplus 1} \vdash \pf{\circ p}{x\oplus 1}
 }
{
\pf{\circ \circ p}{x} \vdash \pf{\circ p}{x\oplus 1}
}
{
	\circ\vdash
}		
 }

\def\y{
	\prova{\x \quad \pf{p}{x} \vdash \pf{p}{x}}
	{
		\pf{p}{x}, \pf{p\to \circ\circ  p}{x} \vdash \pf{\circ p}{x\oplus 1}
	}{\to \vdash}
}
 
 \def\f{
 	\prova{
 		\y	\quad \pf{p}{x\oplus 1}\vdash \pf{p}{x\oplus 1}
 	}
 	{
 		\pf{p}{x},\pf{p}{x\oplus 1}, \pf{p\to \circ\circ p}{x}
 		\vdash 
 		\pf{ p\land \circ p}{x\oplus 1}	
 	}
 	{\vdash \land}
 }
 
 \def\e{
 \prova{
 \f	
 }
{
\pf{p}{x},\pf{\circ p}{x}, \pf{\Box(p\to \circ\circ p)}{0}
\vdash 
\pf{ p\land \circ p}{x\oplus 1}	
}
{\Box\vdash}	
 }

 \def\d{
 	\provas{
 		\e
 	}
 {
 	\pf{p\land \circ p}{x}, \pf{\Box(p\to \circ\circ p)}{0}
 	\vdash 
 	\pf{ p\land \circ p}{x\oplus 1}
 }
{\land\vdash, \land\vdash, contr\vdash}
 }

 \def\c{
\prova{\pf{p}{z}\vdash \pf{p}{z}}{\pf{p\land \circ p}{z}\vdash \pf{p}{z}}{\land\vdash}
 }
 \def\b{
 	\prova{
 	\d	
 	}
 {
 	\pf{p \land \circ p}{0}, \pf{\Box(p\to \circ\circ p)}{0}\vdash \pf{ p\land \circ p}{z}
 }
{IND}
 }

 \def\a{
 	\prova{
 		\b \quad \c
 	}{
 	\pf{p \land \circ p}{0}, \pf{\Box(p\to \circ\circ p)}{0}\vdash \pf{ p}{z}
 }{
 cut
}
 }
 
 \[
 \provas{
 	\a
 }
{
 \pf{p \land \circ p}{0}, \pf{\Box(p\to \circ\circ p)}{0}\vdash \pf{\Box p}{0}
}
{
	\vdash \Box
}
 \]
Observe now that the cut elimination procedure used in Section~\ref{Sect:cut-elim} fails. The induction on the p-formula $\pf{p \land \circ p}{x}$ is not breakable in  simpler  cases, and the cut cannot be eliminated.  
By generalizing the above example, it is possible to exhibit a series of proofs with non-eliminable cuts with unbounded rank.

\subsubsection{On a syntactical consistency proof  for \ltl}

That \ltl{} is consistent is evident by semantical methods. The situation becomes complicated if we want to prove $2_{\ltl}$ consistent by purely syntactical methods.
One possibility is  to give a natural deduction formulation of \ltl{} as we did in~\cite{BarMas:apal}, and prove its consistency by means of a strong normalization theorem (non formalizable in \textbf{PA}).
Another possibility is to give an infinitary formulation, as in~\cite{BaMaAML2004}, where cut elimination and the subformula principle hold.

A third possibility, once that we have a sequent system like  $2_{\ltl}$, is to use ordinal analysis, and try to mimic for $2_{\ltl}$ Gentzen's consistency proof for \textbf{PA} (see, e.g., Takeuti's book~\cite{takeuti:pt}). We leave this  open for further research on the topic.

%% file: pastLTL.tex
\section{Taming the past: only a sketch}\label{sect:past}
As a final step of our journey,  we show how the notion of position can be further generalised, to capture an extension of \ltl{} with operators for  past and  future, that we will call here $\pltl$ (see e.g.~\cite{DGL}). 
Despite its physical and philosophical interest, a logic with both future and past operators did not receive much interest in the context of linear temporal logics for computer science. Temporal logics are investigated especially for their role in the verification of computer systems. There, it makes no sense to assume an unbound past;
on the other hand, \ltl{} with past \emph{and} a beginning instant has been shown no more expressive than \ltl{} (see~\cite{GPSS80,Gab87} where it is proved that any  {\ltl}\ +past property is  equivalent, when  evaluated at the beginning of time, to a suitable \ltl{} formula.)

In this section, to show how our method based on positions may accommodate various notions of modalities, we give a brief outline of how the deductive system of \ltl{} can be extended to model both unlimited future and unlimited past, thus adding to the language the operators: $\PBox$ (always in the past), $\PDiamond$ (sometimes in the past),  and $\Prev$ (at the previous time point).

\subsection{The calculus $2_{\pltl}$}
 For the sake of this section we assume the following definition.

\begin{definition} Given a denumerable set of tokens $\{x_0,x_1,\ldots\}.$, the set of \textit{positions} is the set of all
   triple $<r,S_1,S_2>$ where
   
\begin{itemize}
    \item $r\in \ZZ$ (which will responsible for next/previous steps);
    \item $S_1$ is a finite set of tokens (responsible for future);
    \item $S_2$ is a finite set of tokens (responsible for past).
\end{itemize} 
\end{definition}
\noindent We will use $l$, $m$, $n$ for ranging over $\NN$, and $q$, $r$, $s$ for ranging over $\ZZ$. 
We will need the following notation, for $T$ finite set of tokens:


\begin{enumerate}
\item 
$<r,S_1,S_2> \oplus <m,T> = <r+m, S_1-T, S_2\cup (T-S_1)>$
\item 
$<r,S_1,S_2>\ominus <m,T> = <r-m,S_1\cup (T-S_2), S_2-T>$
\item $<r,S_1 ,S_2> \oplus x = <r,S_1,S_2>\oplus<0,\{x\}>$ 
\item $<r,S_1 ,S_2> \ominus x = <r,S_1,S_2>\ominus<0,\{x\}>$
\item $<r,S_1 ,S_2> \oplus 1 = <r,S_1,S_2>\oplus<1,\emptyset>$
\item $<r,S_1 ,S_2> \ominus 1 = <r,S_1,S_2>\ominus<1,\emptyset>$
\end{enumerate}

The full set of rules of System $2_{\pltl}$ are given in Figure~\ref{fig:ptemprule}.
\begin{figure}[htb!]

\subsubsection*{Identity rules, Structural rules, Propositional rules }
		Those of system $2_{\ltl}$, formulated with the new notion of position.
\subsubsection*{Temporal rules}
\begin{center} 
$\urule %
  {\Gamma,\pf{A}{s\oplus 1}\vdash\Delta } %
  {\Gamma,\pf{\circ A}{s}\vdash\Delta } %
  {\quad \circ \vdash } $ %
\qquad\qquad
$\urule %
  {\Gamma\vdash\pf{A}{s\oplus 1},\Delta} %
  {\Gamma\vdash\pf{\circ A}{s},\Delta } %
  {\quad\vdash\circ} $ %
  \\[3ex]
$\urule %
  {\Gamma,\pf{A}{s\ominus 1}\vdash\Delta } %
  {\Gamma,\pf{\Prev A}{s}\vdash\Delta } %
  {\quad \Prev \vdash } $ %
\qquad\qquad
$\urule %
  {\Gamma\vdash\pf{A}{s\ominus 1},\Delta} %
  {\Gamma\vdash\pf{\Prev A}{s},\Delta } %
  {\quad\vdash\Prev} $ %
  
\end{center}

\begin{center} 
$\urule %
  {\Gamma,\pf{A}{s \oplus <m,T>}\vdash\Delta } %
  {\Gamma,\pf{\Box A}{s}\vdash\Delta } %
  {\quad\Box\vdash } $ %
\qquad \qquad
$\urule %
  {\Gamma\vdash\pf{A}{s \oplus  x},\Delta} %
  {\Gamma\vdash\pf{\Box A}{s},\Delta } %
  {\quad\vdash\Box } $ %
  \\[3ex]
$\urule %
  {\Gamma,\pf{A}{s \ominus <m,T>}\vdash\Delta } %
  {\Gamma,\pf{\PBox A}{s}\vdash\Delta } %
  {\quad\PBox\vdash } $    %
\qquad \qquad
$\urule %
  {\Gamma\vdash\pf{A}{s \ominus  x},\Delta} %
  {\Gamma\vdash\pf{\PBox A}{s},\Delta } %
  {\quad\vdash\PBox } $ %
\end{center}

\begin{center} 
$\urule %
  {\Gamma,\pf{A}{s \oplus  x}\vdash\Delta} %
  { \pf{\Gamma,\Diamond A}{s}\vdash\Delta} %
  {\quad\Diamond\vdash} $  %
\qquad \qquad
$\urule %
  {\Gamma\vdash\pf{A}{s \oplus <m,T>},\Delta } %
  {\Gamma\vdash\pf{\Diamond A}{s},\Delta} %
  {\quad\vdash \Diamond } $ %
  \\[3ex]
$\urule %
  {\Gamma,\pf{A}{s \ominus  x}\vdash\Delta} %
  { \pf{\Gamma,\PDiamond A}{s}\vdash\Delta} %
  {\quad\PDiamond\vdash} $ %
\qquad \qquad
$\urule %
  {\Gamma\vdash\pf{A}{s \ominus <m,T>},\Delta } %
  {\Gamma\vdash\pf{\PDiamond A}{s},\Delta} %
  {\quad\vdash \PDiamond } $ %
\end{center}

\begin{center} 
$
  \urule{
  	\Gamma, \pf{A}{s \oplus x} \vdash  \pf{A}{s \oplus x\oplus 1} , \Delta
  }
  {
  	\Gamma,\pf{A}{s} \vdash \pf{A}{s\oplus <m,T>}, \Delta
  }
  {\rm IND}
  \qquad
  \urule{
  	\Gamma, \pf{A}{s \ominus x} \vdash  \pf{A}{s \ominus x\ominus 1} , \Delta
  }
  {
  	\Gamma,\pf{A}{s} \vdash \pf{A}{s\ominus <m,T>}, \Delta
  }
  {\rm PIND}
  $

\end{center}

\textbf{Constraints}: in $\Diamond\vdash$, $\PDiamond\vdash$, $\vdash\Box$, $\vdash\PBox$, $\rm IND$, $\rm PIND$:  $x\not\in s,\Gamma,\Delta$

\caption{System $2_{\pltl}$}\label{fig:ptemprule}
\end{figure}


\subsection{Soundness}

Given the frame $\ZZ$ of integer numbers,  a
map $a:{\bf \ZZ}\to 2^{At}$, and an  integer $q$, the relation of
satisfiability  by the 
model  ${\ZZ}_a=\langle{\ZZ}, a\rangle$ of a temporal formula $A$ at time $q$
(notation: ${\ZZ}_a\forces_q A$) is defined by induction on the complexity of
$A$ in the standard way. For the sake of completeness we recall the definition in
the case of modalities.

\begin{enumerate}
\item[] $A$ is $\Box B$:\quad   ${\ZZ}_a\forces_q A\ \Leftrightarrow\ {\ZZ}_a\forces_r  B \mbox{\ for all\ } r\ge q;$ 
\item[] $A$ is $\circ  B$:\quad  ${\ZZ}_a\forces_q A\ \Leftrightarrow\ {\ZZ}_a\forces_{q+1}  B.$ 
\item[] $A$ is $\PBox B$:\quad   ${\ZZ}_a\forces_q A\ \Leftrightarrow\ {\ZZ}_a\forces_r  B \mbox{\ for all\ } r\le q;$ 
\item[] $A$ is $\Prev  B$:\quad  ${\ZZ}_a\forces_q A\ \Leftrightarrow\ {\ZZ}_a\forces_{q-1}  B.$ 
\end{enumerate}


It is only routine to prove that.
$$\vdash_{\pltl}A\ \Rightarrow  \ 
{\ZZ}_a \forces_q A \mbox{\ for all\ }a:{\ZZ}\to 2^{At} \mbox{\ and\ } q\in\ZZ.
$$

We fix the frame structure 
$\langle {\ZZ},+,-,0,1,-1\rangle$. Models based on this structure are of the form 
$${\ZZ}_{a,v}=\langle {\ZZ},+,-,0,1-1,,a:T\to{\ZZ},v:{\ZZ}\to 2^{At}\rangle.$$
In ${\ZZ}_{a,v}$ one can assign a value $a(s)\in {\ZZ}$ to each
position $s=<p,S_1,S_2>$: 
$$a(s)=n+\sum_{x\in S_1}a(x) - \sum_{y\in S_2}a(y).$$

\begin{definition} The model ${\ZZ}_{a,v}$ {\em satisfies} the formula $\pf{A}{s}$ (notation:\linebreak
${\bf N}_{a,v}\models{\pf{A}{s}}$) if and only if ${\bf N}_v\forces_{a(s)}A.$
\end{definition}

As for $\ltl$, by this definition the semantics of the position formulas of $\pltl$ is reduced to the standard semantics of \ltl. For instance: 
	\begin{itemize}
	\item ${\ZZ}_{a,v}\models \pf{\Prev A}{s}$ iff 
		${\ZZ}_v\forces_{a(s)}\Prev A$ iff
		${\ZZ}_v\forces_{a(s)-1} A$;
	\item ${\ZZ}_{a,v}\models \pf{\PBox A}{s}$ iff 
		${\ZZ}_v\forces_{a(s)}\PBox A$ iff 
		$\forall n\geq 0. {\ZZ}_v\forces_{a(s)-n} A$;
	\item ${\ZZ}_{a,v}\models \pf{\PDiamond A}{s}$ iff 
		${\ZZ}_v\forces_{a(s)}\PDiamond A$ iff 
		$\exists n\geq 0. {\ZZ}_v\forces_{a(s)-n} A$.
	\end{itemize}
We extend the definition to sequents as:
$$
\ZZ_{a,v} \models \Gamma\vdash\Delta \Leftrightarrow (
\forall \pf{A}{s}\in\Gamma.\NN_{a,v}\models \pf{A}{s}
\Rightarrow \exists \pf{B}{t}\in\Delta.\NN_{a,v}\models \pf{B}{t}
 ).
$$
We finally define:
$$
\Gamma\models_{\pltl}\Delta \Leftrightarrow \forall a,v.  \ZZ_{a,v}, \models \Gamma\vdash\Delta.
$$
	
\begin{lemma} 
	\begin{enumerate}
		\item Let $x\not\in s$, then 
			${\ZZ}_{a,v} \models \PBox \pf{A}{s} \Leftrightarrow 
			\forall n\in \NN. {\NN}_{a[x/n],v} \models   \pf{A}{s\ominus x}$;
		\item Let $x\not\in s$, then ${\ZZ}_{a,v} \models  \pf{\PDiamond A}{s} \Leftrightarrow 
			\exists n\in \NN. {\ZZ}_{a[x/n],v} \models   \pf{A}{s\ominus x}$;
		\item ${\ZZ}_{a,v} \models   \pf{A}{s\ominus t}
			\Leftrightarrow
			{\ZZ}_{a[x/a(t)],v} \models   \pf{A}{s\ominus x}$.
		\end{enumerate}
\end{lemma}

\begin{theorem}[soundness]
	If $\Gamma\vdash\Delta$ is derivable in $2_{\pltl}$, then $\Gamma\models_{2_{\pltl}}\Delta$.
\end{theorem}

The proof of the theorem proceeds, \emph{mutatis mutandis}, as for $2_{\ltl}$: 
define an equivalent system $2^i_{\pltl}$ where the induction rules are substituted with the corresponding axioms, prove the equivalence with $2_{\pltl}$, and establish the soundness of $2^i_{\pltl}$.

\subsection{Examples of derivations}
%
%
\def\BPa{
	\prova{ 
		\pf{A}{<0,\emptyset,\emptyset>}\vdash \pf{A}{<0,\emptyset,\emptyset>}
	}
	{ 
		\pf{A}{<0,\emptyset,\emptyset>}\vdash \pf{\Diamond A}{<0,\{x\},\emptyset>}
	}
	{ \vdash\Diamond } 
}

\def\BPb{
	\prova{ 
		\BPa
	}
	{ 
		\pf{A}{<0,\emptyset,\emptyset>}\vdash \pf{\PBox\Diamond A}{<0,\emptyset,\emptyset>}
	}
	{ \vdash\PBox } 
}

\def\BPbb{
	\prova{ 
		\BPb
	}
	{ 
		\vdash \pf{A\to \PBox\Diamond A}{<0,\emptyset,\emptyset>}
	}
	{ \vdash\to} 
}

\def\PBaa{
	\prova{ 
		\pf{A}{<0,\emptyset,\emptyset>}\vdash \pf{A}{<0,\emptyset,\emptyset>}
	}
	{ 
		\pf{A}{<0,\emptyset,\emptyset>}\vdash \pf{\PDiamond A}{<0,\emptyset,\{x\}>}
	}
	{ \vdash\PDiamond } 
}

\def\PBab{
	\prova{ 
		\PBaa
	}
	{ 
		\pf{A}{<0,\emptyset,\emptyset>}\vdash \pf{\Box\PDiamond A}{<0,\emptyset,\emptyset>}
	}
	{ \vdash\Box } 
}

\def\BPaaa{
	\prova{ 
		\PBab
	}
	{ 
		\vdash \pf{A\to \Box\PDiamond A}{<0,\emptyset,\emptyset>}
	}
	{ \vdash\to} 
}

\def\PCa{
	\prova{ 
		\pf{A}{<0,\emptyset,\emptyset>}\vdash \pf{A}{<0,\emptyset,\emptyset>}
	}
	{ 
		\pf{A}{<0,\emptyset,\emptyset>}\vdash \pf{\Prev A}{<1,\emptyset,\emptyset>}
	}
	{ \vdash\Prev } 
}

\def\PCb{
	\prova{ 
		\PCa
	}
	{ 
		\pf{A}{<0,\emptyset,\emptyset>}\vdash \pf{\circ\Prev A}{<0,\emptyset,\emptyset>}
	}
	{ \vdash\circ } 
}

\def\PCbb{
	\prova{ 
		\PCb
	}
	{ 
		\vdash \pf{A\to \circ\Prev A}{<0,\emptyset,\emptyset>}
	}
	{ \vdash\to} 
}

\def\CPaa{
	\prova{ 
		\pf{A}{<0,\emptyset,\emptyset>}\vdash \pf{A}{<0,\emptyset,\emptyset>}
	}
	{ 
		\pf{A}{<0,\emptyset,\emptyset>}\vdash \pf{\circ A}{<-1,\emptyset,\emptyset>}
	}
	{ \vdash\circ } 
}

\def\CPab{
	\prova{ 
		\CPaa
	}
	{ 
		\pf{A}{<0,\emptyset,\emptyset>}\vdash \pf{\Prev\circ A}{<0,\emptyset,\emptyset>}
	}
	{ \vdash\Prev } 
}
\def\CPaba{
	\prova{ 
		\CPab
	}
	{ 
		\vdash \pf{A\to \Prev\circ A}{<0,\emptyset,\emptyset>}
	}
	{ \vdash\to} 
}

As an example let us show the derivations of the basic axioms of tense logic with past and  future 
(see e.g.~\cite{Burgess}).
\bigskip
 
$$
\BPbb \qquad \BPaaa
$$

\bigskip 
$$
\PCbb \qquad \CPaba
$$

\bigskip\noindent
It is an open question to characterise the exact set of axioms (and thus of models) for which $2_{\pltl}$ is complete. We leave this to further, ongoing research on a wider range of tense logics with operators for  past and future.

%% file: conclusions.tex
\section{Conclusions and Future Work}\label{sec:conclusions}
We have presented several 2-sequent systems for \emph{classical} modal logics, ranging from the basic {\KK} to the more elaborate \ltl{}. The \emph{leitmotiv} of this journey has been the notion of \emph{position}, which allows to fully expose the formal proof-theoretical analogy between modalities and first order quantification. This analogy has been first exploited in our previous papers, that only dealt with classical {\D}, and the intuitionistic, $\Box, \to, \wedge$--fragments (no negation) of {\D}, {\K4}, {\T}, and {\S4}.
We stress again that our approach is motivated by purely proof-theoretical issues, thus different from all other proposals where annotations on formulas (or sequents) are explicitly meant to bring into the syntax the various properties of the Kripke models. It is a feature of our approach that some of the formal combinatorics of such semantical annotations are indeed reconstructed starting from just proof-theoretical motivations. 

In a companion paper we will introduce natural deduction systems for the same logics we presented here. Natural deduction calculi for {\D} and {\K4} are particularly challenging. Indeed, referring to the semantics of Section~\ref{semantics-sequents}, one sees that the mapping from positions to nodes may be undefined on some positions (equivalently, the accessibility relation $\R$ of the Kripke models may be partial). This calls for several constraints on the modal rules, which may be treated by using an ``existence predicate,'' first introduced by Dana Scott in~\cite{Scott1979} to deal with intuitionistic logic with partial terms. In our case, such existence predicated would be applied to positions, thus reinforcing the formal analogy between terms and positions. 

Once we have natural deduction calculi, we may consider their constructive versions, and the lambda-calculi that emerge in that way, by explicitating the proof-terms. Differently from the calculi in~\cite{MM:ComInt:95} (where we did not have a general enough notion of position), positions will not be decorations of terms, but terms themselves, and as such they may be manipulated by other lambda-terms. In a typed version, this will call for dependent types.

A final, interesting topic is to investigate syntactical consistency proofs for \ltl{}, exploiting the relations between rule IND and numerical, first-order induction.